\documentclass[a4paper,onecolumn,11pt,accepted=2020-09-01]{quantumarticle}
\pdfoutput=1
\usepackage[utf8]{inputenc}
\usepackage[english]{babel}
\usepackage[T1]{fontenc}
\usepackage{graphicx, caption}
\usepackage{amsfonts}
\usepackage{amsmath,amsthm,bm}
\usepackage{amssymb}
\usepackage{amsthm}
\usepackage{color}
\usepackage{mathtools}
\usepackage{braket}
\usepackage{bbm}
\usepackage{enumitem}
\usepackage{tikz}
\usepackage[numbers,sort&compress]{natbib}

\usepackage{hyperref}
\usepackage{lipsum}

\newcommand{\id}{{\mathbbm{1}}}

\newcommand\bovermat[2]{%
  \makebox[0pt][l]{$\smash{\overbrace{\phantom{%
    \begin{matrix}#2\end{matrix}}}^{\text{#1}}}$}#2}

\newcommand\trans{^{\mathsf{T}}}

\theoremstyle{plain}
\newtheorem{thm}{Theorem}%[section] 

\newtheorem{defn}[thm]{Definition} % definition numbers are dependent on theorem numbers
\newtheorem{lemma}[thm]{Lemma}
\newtheorem{corollary}[thm]{Corollary}

\newtheoremstyle{remark}% name of the style to be used
  {}% measure of space to leave above the theorem. E.g.: 3pt
  {}% measure of space to leave below the theorem. E.g.: 3pt
  {}% name of font to use in the body of the theorem
  {}% measure of space to indent
  {\itshape}% name of head font
  {.}% punctuation between head and body
  {2ex}% space after theorem head; " " = normal interword space
  {}
\theoremstyle{remark}
\newtheorem*{remark}{Remark}

	\addtocounter{footnote}{-1}%
\begin{document}

\title{Quantum linear network coding for entanglement distribution in restricted architectures}
%\author{Steven Herbert$^{1,2}$ and Niel de Beaudrap$^{1}$}
\author{Niel de Beaudrap}
\email{niel.debeaudrap@gmail.com}
\affiliation{Department of Computer Science, University of Oxford, UK}
\orcid{0000-0001-9549-5146}
\author{Steven Herbert}
\email{sjh227@cam.ac.uk}
\affiliation{Department of Computer Science, University of Oxford, UK}
\affiliation{Riverlane, 1st Floor St Andrews House, 59 St Andrews Street, Cambridge,  UK}
\orcid{0000-0003-0948-5721}
\thanks{{ } \\ { } \\ Both authors contributed equally to the results of this article; the author order is merely alphabetical.}
%\affil[1]{\small{\textit{Department of Computer Science, University of Oxford, UK}}}
%\affil[2]{\small{\textit{Riverlane, 1st Floor St Andrews House, 59 St Andrews Street, Cambridge,  UK}}}
%\date{}
%\maketitle

\begin{abstract}
\noindent In this paper we propose a technique for distributing entanglement in architectures in which interactions between pairs of qubits are constrained to a fixed network $G$.
This allows for two-qubit operations to be performed between qubits which are remote from each other in $G$, through gate teleportation.
We demonstrate how adapting \emph{quantum linear network coding} to this problem of entanglement distribution in a network of qubits can be used to solve the problem of distributing Bell states and GHZ states in parallel, when bottlenecks in $G$ would otherwise force such entangled states to be distributed sequentially.
In particular, we show that by reduction to classical network coding protocols for the $k$-pairs problem or multiple multicast problem in a fixed network $G$, one can distribute entanglement between the transmitters and receivers with a Clifford circuit whose quantum depth is some (typically small and easily computed) constant, which does not depend on the size of $G$, however remote the transmitters and receivers are, or the number of transmitters and receivers.
These results also generalise straightforwardly to qudits of any prime dimension.
We demonstrate our results using a specialised formalism, distinct from and more efficient than the stabiliser formalism, which is likely to be helpful to reason about and prototype such quantum linear network coding circuits.
\end{abstract}

\section{Introduction}

One of the most important problems to solve, in the realisation of quantum algorithms in hardware, is how to map operations onto the architecture.
Scalable architectures for quantum computers are not expected to have all-to-all qubit connectivity:  if we describe the pairs of qubits which may interact directly by the edges of a graph (or ``network'') $G$ whose nodes are qubit labels, then $G$ will not contain all pairs of nodes.
This raises the question of how best to realise two-qubit operations on data stored on pairs of qubits $a,b \in G$ which are not adjacent in $G$.

One solution is to swap qubit states through the network until they are on adjacent nodes~\cite{Cowtan-etal-2019,me,me2}.
An alternative, which is possible when not all qubits in the hardware platform are used to store data, is to distribute entanglement between qubits $a', b' \in G$ which are adjacent to $a$ and $b$ respectively.
This allows a gate between $a$ and $b$ to be performed by teleportation~\cite{newref1}.
Which approach is the more practical will depend on whether it is economical to leave some number of qubits free to use as auxiliary space, but also on how much noise the state is subject to as a result.
The question of which approach will lead to more accumulated noise will be determined in part by how long it takes to realise the chosen approach, in total, over all operations to be performed in a given algorithm.

To reduce the time taken in distributing entanglement for two-qubit operations, we consider how entangled states may be distributed between multiple pairs in parallel.
A direct approach may result in crossing paths in the network $G$, forcing the entangled pairs to be distributed in sequence rather than in parallel.
The issue of crossing paths for transmissions across a network is also potentially an issue in conventional networks.
In that setting, one solution to this problem is \emph{network coding}, in which independent signals in a network may share bandwidth by allowing intermediate nodes to combine their signals in appropriate ways to distribute complete information about each signal across the network.
(A simple illustrative example of this, the ``butterfly network'', is shown in Fig.~\ref{f01}.)
This motivates the idea of using network coding to realise entanglement distribution between multiple pairs of qubits in parallel using similar concepts.

Previous work~\cite{Leung2006,Kobayashi2009,Kobayashi2011,Satoh2012} has shown that when a classical binary linear network code exists for the ``multiple unicast'' problem (the problem of sending signals between $k$ pairs of sources and targets) on a classical network, then there exists a quantum network code to distribute Bell states between each source--target pair in a quantum network of the same connectivity.
However, these results suppose that each ``node'' is a small device, hosting multiple qubits and able to perform arbitrary transformations on them before transmitting onward ``messages'' through the network.
This does not reflect the architecture of many hardware projects to realise quantum computers, in which the ``nodes'' are single qubits, and edges are pairs which may be acted on by a quantum operation (such as a CNOT) rather than a directed communications link~\cite{IBM1,GoogleQC,Rigetti,Intel,IBM2}.
\begin{figure}[!t]
	\centering
	\includegraphics[width=0.73\linewidth]{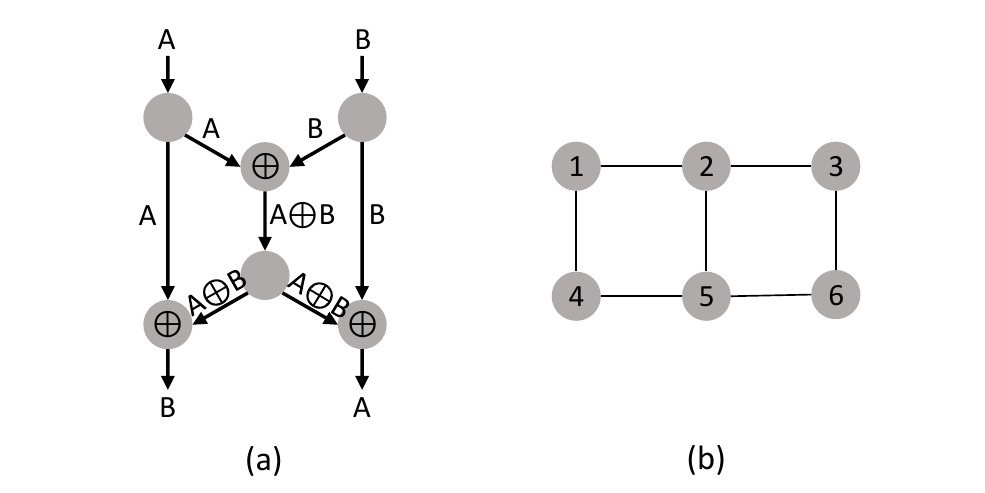}
	\captionsetup{width=0.95\linewidth}
	\caption{\small{(a) Example of network coding over the Butterfly network for input bitstreams ``A'' and ``B'' -- nodes either perform a modulo-2 sum of the incoming bitstreams (when labelled $\oplus$) or fanout the single incoming bitstream otherwise; (b) the Butterfly shown as a (topologically identical) $2 \times 3$ grid, with node order shown by the labelled indices -- as the Butterfly network provides a useful illustrative example for many of the results presented in this paper this ordering is defined and used consistently throughout the paper (for example for qubit ordering).}}
	\label{f01}
\end{figure}

In this article, we describe techniques to translate linear network coding protocols on a directed graph $G$, to circuits --- called here ``QLNC circuits'' --- which involve only preparation of $\ket{0}$ and $\ket{+}$ states, CNOT gates along edges of $G$, unitary Pauli gates (possibly conditioned on classical information, which is communicated without constraints), and measurements of the observables $X$ and $Z$.
Our techniques extend also to the multiple multicast problem, serving to distribute Bell and GHZ states across such a network $G$.

We show that QLNC circuits allow us to distribute entanglement in a circuit whose quantum depth can be bounded from above by simple properties of the architecture network $G$, leading to a modest sized constant for reasonable choices of $G$ (\emph{e.g.},~9 for a square lattice provided no receiver node has four in-coming links).~\label{discn:introSquareLatticeBound}
In particular, the depth is independent of the number of qubit pairs to be entangled, the distance between the nodes in any of the pairs, or the total number of other qubits involved.
In addition to this constant quantum depth, is a dependency on computing classical controls for some of the quantum operations, which is at worst logarithmic in the number of qubits involved.
These are lower circuit depths than can be achieved by realising two-qubit operations by routing~\cite{Cowtan-etal-2019,KMvdG-2019}.
Furthermore, while our results are in some ways similar to what can be achieved with graph states (as described by Hahn~\emph{et al.}~\cite{Hahn}), our techniques are somewhat more versatile and also easier to analyse.
We make these comparisons more precise in Section~\ref{sec:compare}.

As well as describing how network codes can be used to distribute entanglement, in a setting where the nodes in the network represent individual qubits which may interact in pairs along the network, we also note two features of QLNC circuits that make them more versatile than classical linear network coding protocols:
\begin{itemize}
\item
    QLNC circuits can be used to simulate a classical linear network code ``out of order''.
    (Indeed, this is required for our main result, which simulates a linear network code in a depth which may be smaller than the length of the longest transmitter--receiver path in the classical network.)
\item
    Entanglement swapping allows for QLNC circuits to perform entanglement distribution tasks, that \emph{do not} correspond to classical linear network coding protocols --- that is, for networks $G$ in which the corresponding linear network coding problem has no solution.
\end{itemize}
These results hold as a result of using the (unconstrained) classical control to allow a QLNC circuit to simulate a classical linear network code, on a network with more edges than $G$.

Our analysis of QLNC circuits involves a simple computational formalism, which may be of independent interest.
The formalism is similar to classical network coding in its representation of data with time, and allows the easy use of classical network coding results and intuitions to reason about entanglement distribution circuits.
While QLNC circuits are stabiliser circuits, and can be efficiently simulated using the stabiliser formalism, QLNC circuits do not require the full power of the stabiliser formalism to simulate.
This allows us to reason about them more efficiently than is possible even with the stabiliser formalism.
This yields at least a factor $2$ improvement in space and time requirements, and achieves $O(n)$ complexity (without using sparse matrix techniques) to simulate protocols which only involve superpositions of $O(1)$ standard basis states.
These techniques can also be applied to network codes on qudits of prime dimension.

The remainder of the paper is organised as follows.
In Section~\ref{prev} we review existing literature on classical and quantum network coding.
In Section~\ref{prelim} we introduce the QLNC formalism, and present the main results described above.
In Section~\ref{qubitsection} we give the generalisation for prime $d$-level qudit systems.
In Section~\ref{comp} we discuss the computational complexity of simulating circuits using the QLNC formalism, as well as that of discovering linear network codes.
Finally, in Section~\ref{app1}, we include a detailed proof of the Theorem~\ref{mainthm1}, which demonstrates the way in which a QLNC circuit may be regarded as realising a linear network code on an extended network $G' \supseteq G$.

\section{Preliminaries}
\label{prev}

We begin by reviewing the literature on classical and quantum network coding, and an overview of techniques to help the realisation of two-qubit operations in limited architectures.

\subsection{Classical network coding}

\begin{figure}[!t]
	\centering
	\includegraphics[width=0.29\linewidth]{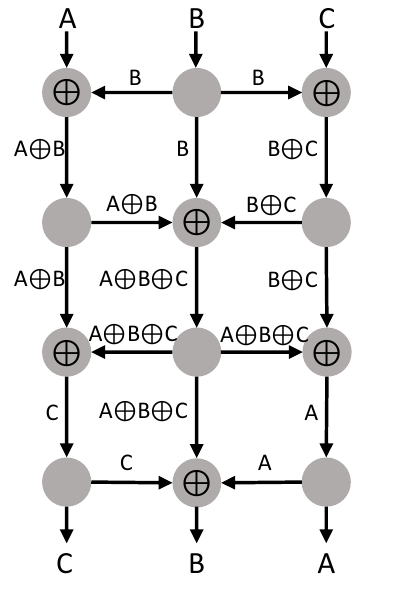}
	\captionsetup{width=0.95\linewidth}
	\caption{\small{Another example of network coding, on a $4\times 3$ grid with three bitstreams ``A'', ``B'' and ``C''.}}
	\label{f01a}
\end{figure}

Network coding, as a means to increase information flow in mesh networks beyond what can be achieved by routing alone, was conceptualised by Ahlswede~\textit{et al}~\cite{Ahlswede2000}.
Rather than simply re-transmit one or more incoming information signals on different onward channels, a network coding protocol allows the nodes in the network to compute some function of its signals (\emph{e.g.},~the \texttt{xor} of bit streams from different incoming links) and to transmit the outcome, in principle ``encoding'' the signals.
The standard example of network coding, providing a simple and clear illustration of the underling principle, is the Butterfly network (Fig.~\ref{f01}), which enables simultaneous transmission between the diagonally opposite corners.
Fig.~\ref{f01a} illustrates a more elaborate network which solves a slightly more complicated signal transmission problem. These examples, which represent a proof of principle of the benefits of network coding, both use binary \textit{linear} network coding -- that is each node can encode its inputs by performing modulo-2 sums. Binary linear network coding provides the basis for the Clifford group QLNCs we address in this paper, however it is worth noting that much of the classical literature considers a more general setting in which the network nodes can encode the input data-streams by performing modulo-$r$ summations (for $r>3$) and / or nonlinear functions. Additionally, these examples are concerned with only one type of network coding task, namely the \textit{multiple unicast} problem (also known as the $k$-pairs problem), in which some number $k \geqslant 1$ of transmitter nodes each send different information streams each to a single, distinct receiver node.
Other problems for which one may consider network coding protocols are the \emph{multicast} and \emph{broadcast} problems (in which a single source node sends the same information to some subset of~the~nodes\,/\,all~nodes in the network respectively), and the \textit{multiple multicast} problem (in which multiple transmitters send different information streams to different subsets of the other nodes).
%These simple examples represent a proof of principle of the benefits of network coding.

The advantage of network coding is most important in the case that the network $G$ has edges which are all directed (as illustrated in the examples of Figs.~\ref{f01} and~\ref{f01a}).
In the case of directed networks, it is always possible to contrive situations in which network coding can yield an unbounded increase in information throughput (for a $k$-pairs example see Fig.~\ref{f025}).
\begin{figure}[!t]
	\centering
	\includegraphics[width=0.73\linewidth]{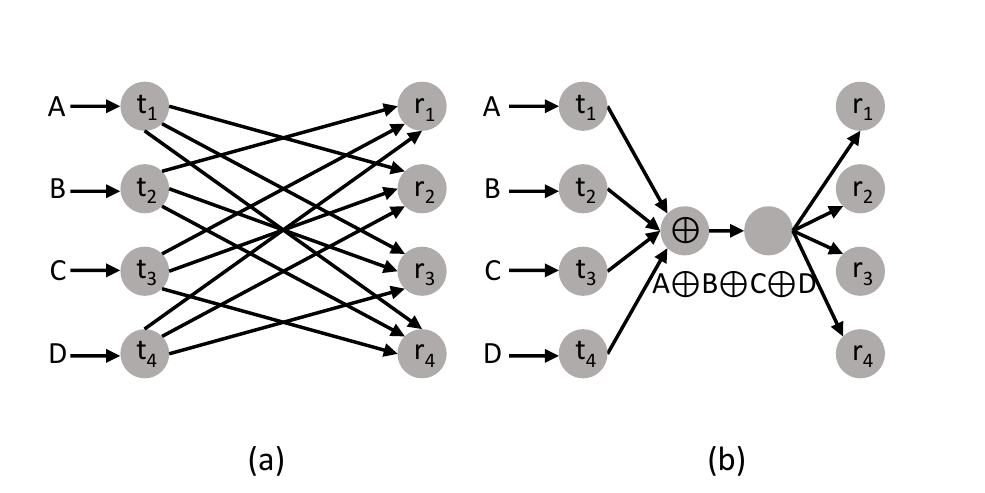}
	\captionsetup{width=0.95\linewidth}
	\caption{\small{An example of directed network in which network coding can yield an arbitrary speed-up in the $k$-pairs setting. The network is a directed graph, consisting of transmitters on the left hand side, and receivers on the right-hand side. Each receiver is paired with the transmitter horizontally left (as shown by the indexed ``t''s and ``r''s). The network consists of two components, a bipartite graph between the transmitters and receivers, with direct links t$_i$-r$_i$ missing, shown in (a); and all of the transmitters connected to all of the receivers through a single directed link, shown in (b). Clearly without network coding all of the transmitter-receiver pairs will have to share the link in (b), and the links in (a) will be useless, however with network coding each of the transmitters can broadcast its bitstream to each output, and the left-most of the central nodes in (b) can perform a modulo-2 sum of all its inputs and forward the result, and the right-most of the central nodes in (b) simply broadcasts this to each receiver. So it follows that each receiver receives 4 bitstreams -- the modulo-2 sum of all the transmissions, via the central nodes, and the bitstreams from all transmitters other than its pair, thus can perform a modulo-2 sum to resolve the bitstream from its paired transmitter. That is, for example, r$_1$ receives B, C and D directly from t$_2$, t$_3$ and t$_4$ respectively, as well as A$\oplus$B$\oplus$C$\oplus$D from the central nodes in (b), and can thus perform the modulo-2 sum of all its inputs A$\oplus$B$\oplus$C$\oplus$D$\oplus$B$\oplus$C$\oplus$D=A, as required. It can easily be appreciated that this construction can extend to any number of transmitter-receiver pairs.}}
	\label{f025}
	%An example of directed network in which network coding can yield an arbitrary speed-up in the $k$-pairs setting. The $k$ transmitters occupy the region of the graph shown in the left-hand grey oval, and the corresponding $k$ receivers occupy the region of the graph shown in the right-hand oval. The network is composed of the two components shown in (a) and (b): (a) represents a complete bipartite structure with the exception of direct links from each transmitter to its receiver, and (b) shows all of the transmitters being connected to a single node, which is connected to a second node, that fans out to all receivers, all of the links have unit bit-rate. Therefore each transmitter is only connected to its receiver via the central pair of nodes in (b) and so without using network coding the bit-rate for any one transmitter-receiver pair can be arbitrarily slow. Whereas with network coding each transmitter can send its bitstream over all of its outgoing channels, and the left-hand central node in (b) can perform a modulo-2 addition of all the incoming bitstreams, forward this to the right-hand central node, which then fans-out this bitstream to all of the receivers. In this way each transmitter-receiver pair can communicate its bitstream at unit bit-rate. For example receiver 1 receives the bitstreams from transmitters $2 \cdots k$ directly (i.e., over the links in (a)), and the modulo-2 sum of all bitstreams from the right-hand central node in (b). So it follows, that performing a modulo-2 sum in receiver 1 means that the bitstream from transmitter 1 irecovered,  required
\end{figure}
However, in many practical contexts, the available communication channels are bidirectional.
For such networks, it is often not clear that network coding will yield any benefits at all.
For the broadcast setting, it has been proven that there is no benefit to the application of network coding over standard routing~\cite{Li2004a}.
For tasks of transmitting long information streams in undirected networks, other techniques than network coding appear to be competitive.
For instance, \emph{fractional routing} involves dividing up a single bitstream and forwarding it along different routes, storing them locally in between rounds of use of the network.
Fig.~\ref{f03} illustrates how fractional routing can achieve the same asymptotic throughput as network coding in the Butterfly Network.

The \textit{multiple unicast conjecture} posits that there is no benefit to the application of network coding over standard routing for multiple unicast channels, if fractional routing is possible \cite{Li2004b}. 
While the multiple unicast conjecture remains unproven, the improvement possible by using network coding has been upper-bounded to typically low factors for various restricted multiple unicast settings \cite{Cai2015}. 
This rather sets the tone for the other settings considered, with an upper bound equal to two on the factor improvement over routing achievable by applying network coding being proven for the multicast and multiple multicast settings \cite{Li2009}. Table~\ref{tab1} summarises the benefits of network coding in various settings.

\begin{figure}[!t]
	\centering
	\includegraphics[width=\linewidth]{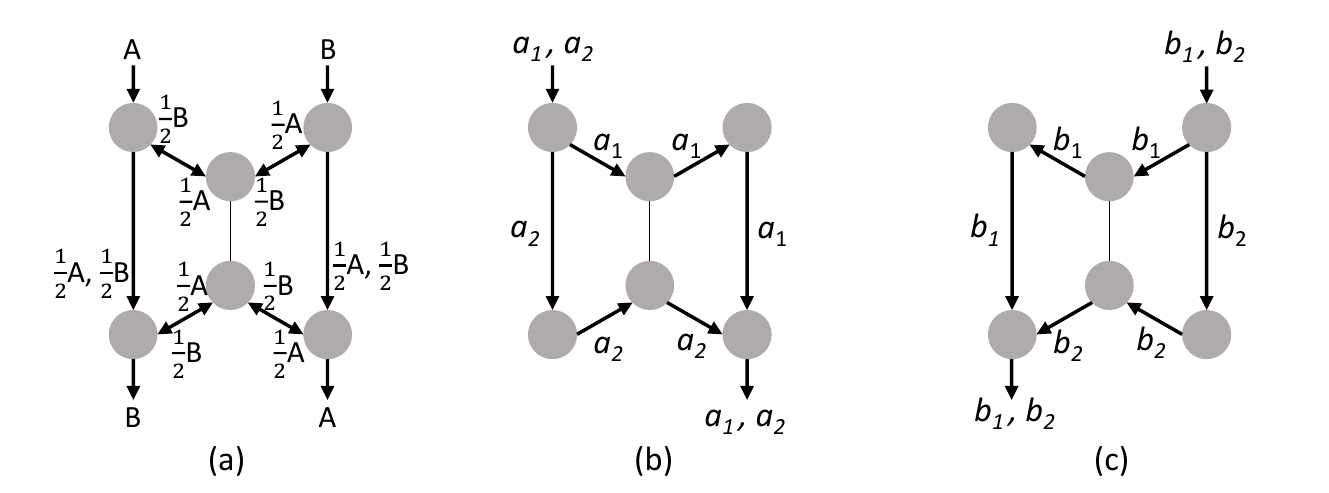}
	\captionsetup{width=0.95\linewidth}
	\caption{\small{Demonstration of achieving the same throughout on the Butterfly as network as network coding by using fractional routing instead. This is achieved by splitting each of bitstreams ``A'' and ``B'' into halves, and forwarding half on each link, as shown in (a). That is (for example) let A consist of two bits $a_1$ and $a_2$ and likewise B two bits $b_1$ and $b_2$. In the first time interval the two bits of A are forwarded on different routes, as shown in (b), and then likewise for the bits of B, shown in (c). Thus time-sharing is used to achieve the fractional routing, and A and B can each forward two bits in a total of two time intervals, which corresponds to the same bit-rate as is achieved using network coding, as shown in Fig.~\ref{f01}.}}
	\label{f03}
\end{figure}

\begin{table}[t]
    \bigskip
	\centering
	\begin{tabular}{ l c c c c c} 
	& \textbf{Broadcast} & \textbf{Multicast} & \textbf{Multiple multicast}& \textbf{Multiple unicast}\\
		\hline\hline 
		\textbf{Directed} & $\infty$ & $\infty$ & $\infty$ & $\infty$ \\
		\textbf{Undirected} & 1 & $\leq 2$ & $\leq 2$ & 1 (conjectured)\\
	\end{tabular}
	\captionsetup{width=0.95\linewidth}
	\caption{\small{Maximum factor increase in information throughput using network using, for various network and information transfer types.}}
	\label{tab1}
\end{table}

\subsection{Quantum network coding}

The concept of network coding has been adapted to quantum state transmission~\cite{Leung2006}, and then to entanglement swapping~\cite{Kobayashi2009, Kobayashi2011, Satoh2012} in quantum communications networks.
Because of the limitation imposed by the no-cloning theorem, the $k$-pairs problem (or for entanglement swapping, the problem of distributing entanglement between $k$ different transmitter--receiver pairs) is typically the problem studied in the quantum case. It has been shown that any situation in which a classical network code for multiple unicast exists, then there is also a quantum network code for entanglement swapping~\cite{Leung2006, Kobayashi2009, Kobayashi2011}. These results include quantum generalisations of both linear and non-linear network codes.
It is with the former that we are concerned in this article, and Satoh~\textit{et al.} provide a very good visual demonstration of the correspondence between classical and quantum linear network coding for the case of the Butterfly graph~\cite{Satoh2012}.
In the case of ``classically assisted'' quantum linear network coding, in which classical communication is less constrained than quantum communication, de Beaudrap and Roetteler~\cite{deBeaudrap2014} show how quantum network coding can be described as an instance of measurement-based quantum computation involving $X$ observable measurements to remove correlations between the input states and the states of qubits (or qudits) at interior nodes. 

One feature which is common for these existing pieces of research is that they consider quantum networks which are in the same essential form as classical telecommunications networks: nodes which have more than one qubit of internal memory (with negligible latency operations), which are connected to each other by channels with significant latency.
This model is appropriate for entanglement distribution in quantum communication networks, but for entanglement distribution in quantum computers it may be relevant to consider a finer scale model in which each node is itself a single qubit.
Note that in this setting, fractional routing is made more complicated by the inability to store and transmit information without resetting the state of the qubit, making the multiple unicast less plausible.
(In the case that the ``information stream'' consists of a single Bell state between each of the $k$ transmitter/receiver pairs, fractional coding loses its meaning entirely.)

\subsection{Other approaches to realise two-qubit operations in limited architectures}
\label{sec:compare}

While we consider the problem of distributing entanglement in limited quantum architectures, this is not the only approach to the problem of realising two-qubit operations between remote qubit pairs.
We consider below other approaches to this problem

\subsubsection{Realising two-qubit operations via mapping/routing}

One way in which two-qubit operations can be realised between qubits is simply by moving the data stored by these qubits to adjacent nodes, \emph{e.g.},~using logical SWAP operations to exchange the data held by adjacent qubits.
We may then consider the way that such a circuit of SWAP gates (or several such exchanges of qubits) can be decomposed into more primitive gates~\cite{ZPW-2018,Cowtan-etal-2019}.
More generally, we may consider how to decompose a single ``long-distance'' operation (such as a CNOT) between remote qubits, into primitive gates consisting of single-qubit gates on adjacent qubits~\cite{KMvdG-2019}.

These mapping/routing results are applicable to the NISQ setting, \emph{i.e.},~the near-term prospect of hardware platforms in which all or nearly all of the qubits will store data which ideally is not to be lost or disturbed owing to the scarcity of memory resources.
They give rise to unitary circuits, whose depth must scale at least as the distance between the pair of qubits on which we want to perform a two-qubit operation.
However, it seems plausible that the parity-map techniques of Ref.~\cite{KMvdG-2019} may be interpreted in terms of linear network codes.
This might allow their techniques for finding suitable CNOT circuits (in the NISQ setting), to be combined with our techniques for distributing entanglement (in a setting where memory is less scarce).

\subsubsection{Sequential distribution of Bell pairs}

Our approach is to consider how multiple Bell pairs may be distributed through a quantum hardware platform in spite of ``bottlenecks'' in the network of the architecture, in a way that is independent of the distance between the qubits to be entangled.
Note that individual Bell pairs can be distributed in constant depth as well, by taking advantage of the concept of entanglement swapping (a~concept which implicitly underlies our techniques as well).

In (otherwise idealised) quantum hardware with paralellisable two-qubit interactions limited to a connected, undirected network $G$, we may distribute entanglement between any pair of qubits $q$ and $q'$ by first preparing a long chain of entangled qubits, and ``measuring out'' all intermediate qubits (essentially using what we call  ``qubit termination'' below), in constant time.
It suffices to consider a chain $q_0, q_1, \ldots, q_\ell$ of qubits with $q$ and $q'$ as endpoints, and to perform the following:
\begin{enumerate}[itemsep=0ex]
\item
    Prepare each $q_j$ for $j$ even in the state $\ket{+}$, and the remaining qubits in the state $\ket{0}$.
\item
    Perform a CNOT from qubit $q_j$ to qubit $q_{j{-}1}$ for each even $j > 0$.
\item
    Perform a CNOT from qubit $q_j$ to qubit $q_{j{+}1}$ for each even $j < \ell$.
\item
    Measure the $X$ observable on each $q_j$ for $0 \!<\! j \!<\! \ell$ even (recording the outcome $s_j = \pm 1$); and measure the $Z$ observable on each $q_j$ for $j \!<\! \ell$ odd (discarding the outcome and assigning $s_j = +1$).
\item
    If $\prod_j s_j = -1$, perform a $Z$ operation on either $q_0$ or $q_\ell$ (not both).
\end{enumerate}
The value of the product $\prod_j s_j$ can be evaluated by a simple (classical) circuit of depth $O(\log \ell)$, and only determines the final single-qubit operation which determines whether the prepared state is $\ket{\Phi^+}$ or $\ket{\Phi^-}$ on $\{q,q'\}$; the rest of the procedure is evidently realisable by a quantum circuit with a small depth, independent of $\ell$. 

To distribute Bell states between $k$ pairs of qubits, it clearly suffices to perform the above procedure $k$ times in sequence, independently of whether the chains involved cross one another.
(Furthermore, any pairs of qubits whose chains do not cross in $G$ can be processed in parallel.)
As the final corrections can be performed in parallel, the total depth of this procedure is then at most $4k+1$, regardless of the distance between the nodes or the size of $G$.

One of our main results (Theorem~\ref{thm:constdepth} on page~\pageref{thm:constdepth}) is to demonstrate conditions under which we may use a QLNC circuit to simulate a classical linear network coding protocol, in ``essentially constant'' depth --- that is, using a number of rounds of quantum operations which is independent of the size of the network, or the distance between transmitters and receivers.
Thus, for sufficiently large $k$, our techniques will distribute the same entangled states in parallel, with a lower depth of quantum operations than distributing the same entanglement sequentially.

\subsubsection{Distribution of entanglement via graph states}

Our techniques yield results that are in some ways similar to results involving graph states~\cite{graphstate}.
We describe some of these here.

In the work by de~Beaudrap and Roetteler~\cite{deBeaudrap2014}, linear network codes give rise to measurement-based procedures involving graph states (which differ from, but are in some cases very similar to, the coding network itself).
The connection to measurement-based quantum computing informed our results, and in particular our techniques feature both measurements and the depth-reduction for which measurement-based computing is known.
However, as our results rest upon unitary operations performed on a network in which each node is a single qubit, the results of Ref.~\cite{deBeaudrap2014} do not directly apply.

More intriguingly, Hahn~\textit{et al.}~\cite{Hahn} have shown how entanglement can be ``routed'' from an initial graph state using transformations of graph states by local complementations.
Graph states can be prepared in depth equal to the edge-chromatic number of the graph (\emph{i.e.},~as with our results, with depth independent of the size of the distances between the qubits involved).
In this sense they represent a better-known way to address the problem of shallow-depth multi-party entanglement distribution in restricted architectures.
Our results differ from those of Hahn~\textit{et al.}~\cite{Hahn} in that we are able to avoid using the sophisticated technique of local complementation of graph states, instead reducing the problem of entanglement distribution to the somewhat more easily grasped subject of linear network coding, which has also been well-studied in the context of information technologies.
There are also entanglement distribution tasks which cannot be achieved by local transformations of graph states, which can be achieved through our techniques: see Section~\ref{sec:gphState-separatingExample}.

\section{Quantum Linear Network Coding circuits}
\label{prelim}

In this Section, we describe techniques to distribute entanglement in architectures where the pairs of qubits which can interact are restricted to some graph $G$.
Our results involve stabiliser circuits which in a sense simulate a linear network coding protocol on $G$ in order to distribute entanglement, given that the ``nodes'' are single qubits and the ``channels'' consist just of whether or not a CNOT operation is applied.
For this reason, we call these circuits \emph{quantum linear network coding circuits} --- or henceforth, QLNC circuits.

We demonstrate below how to simulate a particular classical linear network code using a QLNC circuit, and how doing so can be used to distribute Bell states in parallel by reducing this task to the $k$-pairs problem.
More generally, we show that the same techniques may be used to distribute GHZ states of various sizes by reducing this task to the multiple multicast problem.
We also demonstrate the way in which QLNC circuits allow us to find solutions which somewhat extend what can be achieved by reduction to the $k$-pairs or multiple multicast problems.
To help this exposition, we introduce a formalism to describe the effect of QLNC circuits as a class of quantum circuits, independent of the application of entanglement distribution.

\subsection{A first sketch of QLNC circuits}

Consider a network $G$ with $k$ transmitters $T = \{t_1, \ldots, t_k\}$ and $k$ receivers $R = \{r_1, \ldots, r_k\}$, where we wish to distribute a Bell pair $\ket{\Phi^+}$ between each pair $(t_j, r_j)$, $1 \le j \le k$.
The simplest application of our techniques is to reduce this problem to the existence of a linear network coding solution to the corresponding $k$ pairs problem on $G$, which we may describe by a subgraph $G'$ (omitting edges not required by the protocol) whose edges are given directions by the coding protocol.\footnote{%
    Note that this is not an easy problem in general: see Section~\ref{comp}. 
}
In particular, our results apply to linear network codes in which, specifically, all nodes with output channels send the same message (consisting of the sum modulo~2 of its inputs) on each of its output channels.

We suppose that classical information may be transmitted freely, without being constrained to the network.
While there will be non-trivial costs associated with communicating and computing with classical information, it is reasonable to suppose that the control system(s) governing the quantum architecture can perform such tasks, without being subject to the restrictions involved in the interactions between qubits.

\subsubsection{Directly simulating classical linear network codes}

Given a linear network code as above, to send a standard basis state from each transmitter to their respective receiver would be straightforward, using a circuit of CNOT gates to simulate the network code.
It would suffice to simply initialise all qubits to $\ket 0$, and at each node, compute the message that the node should transmit by using CNOT gates (oriented along the directed edge) to compute the parity of its incoming message(s) at the corresponding qubit. Fig.~\ref{f001a} illustrates this in the case of the Butterfly network.

\begin{figure}[!t]
	\centering
	\includegraphics[width=0.66\linewidth]{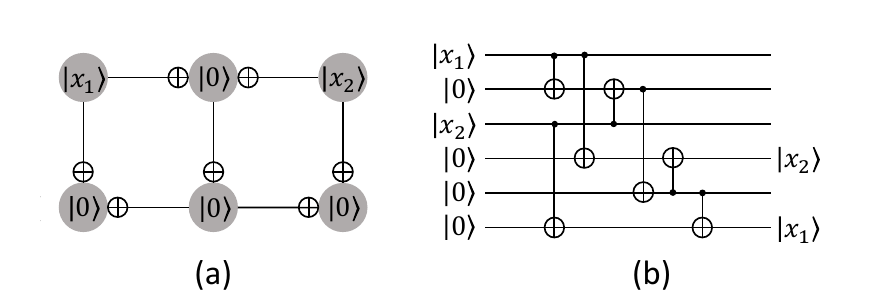}
	\captionsetup{width=0.95\linewidth}
	\caption{\small{Sending computational basis states $x_1$ and $x_2$ over a butterfly network in which each vertex is a qubit, and each edge in a CNOT gate, shown in (a) -- the order in which the CNOT gates are performed is given in the circuit, shown in (b). The grey numbers shown next to vertices in (a) are vertex indices, such that 1 -- 6 is upper-most to lower-most in (b).}}
	\label{f001a}
\end{figure}

To transmit  Bell  pairs, requires additional operations: if the  qubits at the transmitter nodes do not initially start in the standard basis, the procedure described above will yield states in which the transmitters and receivers are entangled with the intermediate nodes.
This elaborated procedure is illustrated in Fig.~\ref{f001}. Following Refs.~\cite{Kobayashi2009, Kobayashi2011,deBeaudrap2014}, we adapt classical network coding protocols by preparing the transmitter states in the $\ket{+}$ state (conceived of as a uniform superposition over standard basis states), and performing $X$ observable measurements (\emph{i.e.},~measurements in the  $\ket{\pm}$ or ``Fourier'' basis) to disentangle the intermediary qubits while leaving them in (joint) superpositions of the standard basis.
These measurements yield outcomes $\pm 1$.
The $+1$ outcome represents a successful disentangling operation, erasing any local distinctions between possible standard basis states without introducing any relative phases.
The $-1$ outcome represents a disentangling operation requiring further work, as a relative phase has been introduced between the possible standard basis states locally.
By conceiving of the state of the qubit as being the parity of some (undetermined) bit-values originating at the transmitters, one may show that it is possible to correct the induced phase by performing $Z$ operations on an (easily determined) subset of the transmitters or receivers.
We refer to this procedure, of measuring a qubit with the $X$ observable and performing appropriate $Z$ corrections, as \emph{termination} of the qubit.
By considering the state of the qubits in Fig.~\ref{f001}(b) after the Hadamard gates simply as a superposition $\tfrac{1}{2} \sum_{a,b} \ket{a,0,b,0,0,0}$ for $a,b \in \{0,1\}$, it is easy to show that the final state after the measurements and classically controlled $Z$ operations is $\tfrac{1}{2} \sum_{a,b} \ket{a,\!\:\cdot\!\;,b,b,\!\:\cdot\!\;,a} = \ket{\Phi^+}_{1,6} \ket{\Phi^+}_{3,4}$, using dots as place-holders for the measured qubits $2$ and $5$.

\begin{figure}[!t]
	\centering
	\includegraphics[width=0.66\linewidth]{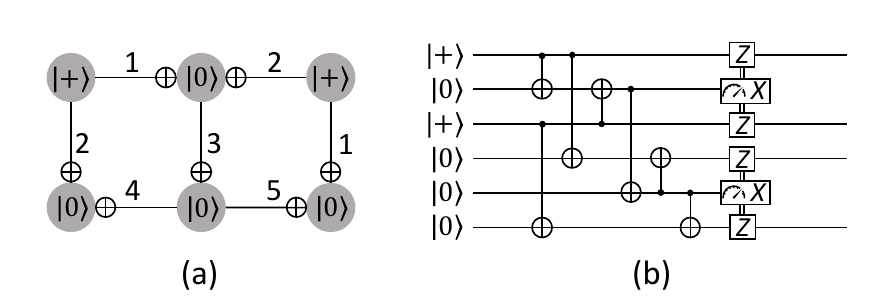}
	\captionsetup{width=0.95\linewidth}
	\caption{\small{Example of performing the Butterfly on a single qubit: (a) shows the order of edges; (b) shows the quantum circuit diagram -- note that the final two layers consisting of Hadamard gates and measurements on qubits 2 and 5, and classically controlled Pauli-$Z$ gates on the other four qubits are necessary for the ``termination'' of qubits 2 and 5, which do not appear in the final desired entangled state. We discuss the general process of termination in full in due course. Nodes are indexed as in Fig.~\ref{f001a}.}}
	\label{f001}
\end{figure}

\subsubsection{Simulating classical linear network codes ``out of order''}
\label{outoforder}

For the application of distributing entanglement, QLNC circuits may simulate linear network coding protocols in other ways than sequential evaluation.
As a fixed entangled state represents a non-local correlation rather than information as such, it suffices to perform operations which establish the necessary correlations between the involved parties.
This principle applies to the simulation of the network coding protocol itself, as well as to the eventual outcome of the entanglement distribution procedure.
For instance: the role of a node with exactly one output channel in our setting is to establish (for each possible standard basis state) a specific correlation between the parities of the qubits of the nodes which are adjacent to it: specifically, that the total parity should be zero.
These correlations may be established without simulating the transmissions of the classical network code in their usual order.

Fig.~\ref{f1} illustrates a mild example of how a QLNC circuit may simulate a classical network protocol (again on the Butterfly network), performing the operations ``out of order''.
\begin{figure}[!t]
	\centering
	\includegraphics[width=0.66\linewidth]{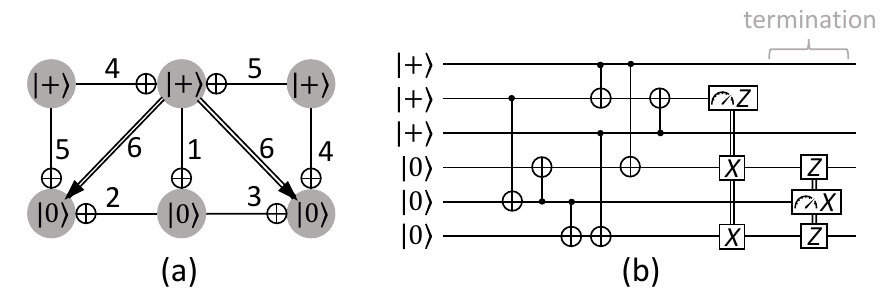}
	\captionsetup{width=0.95\linewidth}
	\caption{\small{The Butterfly performed out of order, as illustrated graphically in (a), with the measurement of qubit 2 performed immediately prior to the classical control; (b) shows the corresponding quantum circuit, and exhibits a good example of the termination process, as described in detail later on. Nodes are indexed as in Fig.~\ref{f001a}.}}
	\label{f1}
\end{figure}
In this case, the correlation between the values of the qubits $1$, $3$, and $5$ (that their projections onto the standard basis should have even total parity, before the disentangling measurement on $5$) is established by attempting to project the qubit $2$ onto the state $\ket{0}$, via a $Z$ observable measurement. In the case that the outcome is instead $\ket{1}$, we must correct any receiver nodes which would be affected by this, by performing (classically conditioned) $X$ operations (represented by the doubled operations, and performed at the sixth time-step).
Again, by considering the state of the qubits in Fig.~\ref{f1}(b) after the Hadamard gates simply as a superposition  $\smash{\tfrac{1}{2\sqrt 2}} \sum_{a,b,z} \ket{a,z,b,0,0,0}$ for $a,b,z \in \{0,1\}$, it is easy to show that the state immediately prior to the measurement of qubit $2$ is $\smash{\tfrac{1}{2\sqrt 2}} \sum_{a,b,z} \ket{a,(z{\oplus}a{\oplus}b),b,(a{\oplus}z),z,(b{\oplus}z)}$, and that projecting qubit $2$ onto the state $\ket{0}$ projects onto those terms for which $z = a \oplus b$.
(Projection onto $1$ projects onto those terms for which $z \oplus 1 = a \oplus b$, and we may correct for this simply by performing an $X$ operation on each receiver whose state depends on the index $z$.)
It is then easy to verify, as with Fig.~\ref{f001}, that the resulting circuit prepares the state $\ket{\Phi^+}_{1,6} \ket{\Phi^+}_{3,4}$. 

One insight is that the freedom to communicate classical information outside of the network allows QLNC circuits to represent a linear network code on a larger network than the network $G$ which governs the two-qubit interactions --- with the qubits as nodes, and both the CNOT~gates\,/\, classically controlled $X$ gates as directed edges.
We will formalise this insight in Section~\ref{main}.

\subsubsection{A separation between QLNC circuits and local transformations of graph states}
\label{sec:gphState-separatingExample}

There are entanglement distribution tasks which can be achieved using QLNC circuits, which cannot be achieved using local transformations of graph states.
Fig.~\ref{novfigref} demonstrates a QLNC circuit on a simple network, whose effect is to prepare a four-qubit GHZ state on the nodes of degree~1.
\begin{figure}[!t]
	\centering
	\includegraphics[width=0.66\linewidth]{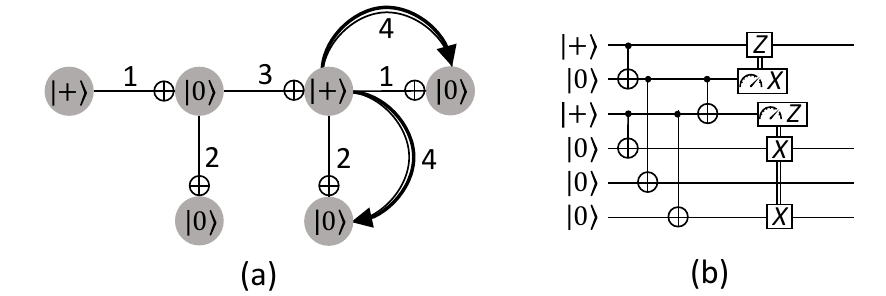}
	\captionsetup{width=0.95\linewidth}
	\caption{\small{An example of an entanglement distribution task separating QLNC circuits from local transformations of graph states.
	The qubits are numbered 1,2,3,4 (left to right) along the top row, and 5, 6 (left to right) along the bottom row.
	Qubit 2 is terminated and qubit 3 is measured (followed by the required classical correction), leaving a four-qubit GHZ state.}}
	\label{novfigref}
\end{figure}
An exhaustive search of the local complementation orbit, including measurements, revealed that the four-qubit GHZ state could not be reached by local Clifford operations and measurements if a graph-state was prepared over the same graph. (We provide the code for this exhaustive search~\cite{ourcode}, which was written specifically for this example but could in principle be adapted for any single network).

While we do not make any formal claim to this effect, the existence of this example leads us to believe that our techniques may yield solutions for entanglement distribution in larger-scale networks and for a variety of entanglement distribution tasks, where it may be difficult or impossible to find a procedure to do so by manipulation of graph states.

%\textbf{\small\sffamily [add any qualification of constant vertex degree, edge-chromatic number, sufficiently large threshold for $k$, etc.\ which may be required]}

\subsection{The QLNC formalism}
\label{main}
Our main objective is to demonstrate how to simulate a classical linear network code to solve a multiple multicast problem on a network $G$, using a QLNC circuit of constant depth, to distribute Bell states and GHZ states between disjoint subsets of qubits located at the nodes of an network $G$ describing the interaction graph of some hardware platform.
To do so, it will be helpful to introduce a simulation technique (which we call the ``QLNC formalism'') to describe the evolution of a set of qubits in a QLNC circuit.

QLNC circuits are stabiliser circuits, by construction.
Indeed, as the only operations which they involve are preparations of $\ket{0}$ and $\ket{+}$ states, CNOT gates, $X$ and $Z$ observable measurements, and unitary $X$ and $Z$ gates conditioned on measurement outcomes, they do not even generate the Clifford group.
For this reason, one might consider using the stabiliser formalism to simulate a QLNC circuit.
(Indeed, the QLNC formalism described below uses operations similar to those of the simulating stabiliser formalism, in that they involve transformations and row-reduction of matrices over $\mathbb Z_2$.) 
The QLNC formalism differs from the stabiliser formalism by expressing states implicitly as superpositions over standard basis states, essentially as a special case of that of Dehaene and de Moor~\cite[Theorem 5~(ii)]{dehaene}.
This renders certain features of the correlations between states immediately obvious: \emph{e.g.},~not even a small amount of pre-processing (such as that required by the stabiliser formalism) is needed to establish the state of any single-qubit state which factorises from the rest.
This alternative representation more transparently represents the qualities of the state which are important to simulate network coding: for this reason, it proves to be a somewhat more efficient method than the stabiliser formalism for this purpose.

\subsubsection{Parity formula states}

In the QLNC formalism, the global state is represented by an assignment of Boolean formulae $f_j(a)$, where $a = (a_1,a_2,\ldots,a_N)$ to each qubit $1 \le j \le n$ in the network $G$.
We call each formula $f_j(a)$ a \emph{node formula} or a \emph{qubit formula}.
Here,
\begin{equation}
  f_j(a)
 \,=\,
  c_{j,0} + c_{j,1} a_1 + \cdots + c_{j,N} a_N\,,
\end{equation}
for some explicit coefficients $c_{j,0},\!\; c_{j,1},\!\;  \ldots,\!\; c_{j,N} \in \{0,1\}$, and where addition is taken modulo $2$ (\emph{i.e.},~each function $f_j(a)$ computes the parity of $c_{j,0}$ and some given subset of its arguments).
These arguments consist of some number of formal indeterminates $a_1, \ldots, a_N$, which we may interpret as variables which may take Boolean values but where those values are as yet undetermined.
We require that, together the the vector $\mathbf e_0 = {[1\:\:0\:\:\cdots\:\:0]\trans}$, the vectors $\{{\mathbf c}_1, {\mathbf c}_2, \ldots, {\mathbf c}_n\} \subseteq \mathbb Z_2^{N+1}$ for ${{\mathbf c}_j \!\!\;= [c_{j,0}\:\:c_{j,1}\:\:\cdots\:\:c_{j,N}]\trans}$ span a set of $2^{N+1}$ vectors.
In particular, each indeterminate $a_h$ must occur in \emph{some} qubit formula $f_j(a)$.
The state also has an associated phase formula $\phi(a)$ of the form
\begin{equation}
  \phi \,=\, p_{0} + p_{1} a_1 + \cdots + p_{N} a_N\,,
\end{equation}
for some coefficients $p_0, \ldots, p_N \in \mathbb{Z}_2$.
Given such a phase formula $\varsigma$ and node-formulas $f_1, f_2, \ldots, f_n$ for a network $G$ of $n$ nodes, the global state of the system is given by
\begin{equation}
  \label{eqn:parityFormulaExpan}
    \frac{1}{\sqrt{2^N}} \!
      \sum_{x \in \{0,1\}^N} \!\!
        (-1)^{\phi(x)}\,
      \ket{f_1(x)} \otimes \ket{f_2(x)} \otimes \cdots \otimes \ket{f_n(x)}
\end{equation}~\\[-2ex]
where $x = (x_1, x_2, \ldots, x_N)$.
That is: the phase formula $\phi(a)$ and node-formulae $f_j(a)$ stand for an explicit  superposition over the standard basis, ranging over all possible substitutions of Boolean strings $x \in \{0,1\}^N$ to the indeterminates $a_1, \ldots, a_N$, and where in particular $\phi(a)$ determines the relative phases.

\begin{defn}
A \emph{parity formula state} is a $n$-qubit state for $n \ge 1$ as expressed in \eqref{eqn:parityFormulaExpan}, where $\phi$ and $f_j$ for $1 \le j \le n$ are (not necessarily homogeneous) linear functions of $N \ge 0$ indeterminates, and where the functions $f_j(a)$ together with the constant function $e_0(a) = 1$ span a set of $2^{N+1}$ functions.
\end{defn}

It will be convenient to consider a representation of parity function states in terms of an $(N+1) \times (n+1)$ matrix $C$ and a separate column vector $\mathbf p$, where $\mathbf p = {[p_0\:\:p_1\:\:\cdots\:\:p_N]\trans}$, and where the columns of $C$ (indexed from $0$ to $n$) consist of the vector $\mathbf e_0$ and the columns $\mathbf c_1, \ldots , \mathbf c_{n}$.

\begin{defn}
    A parity function matrix $C$ for a $n$-qubit state is an $(N{+}1)\times(n{+}1)$ matrix for some $N \ge 0$, of the form $C = \bigl[\mathbf e_0\:\:\mathbf c_1\:\:\cdots\:\:\mathbf c_{n+1}\bigr]$ of rank $N+1$.
    A parity function tableau is a matrix $T = \bigl[\:\! C\,\big\vert\, \mathbf p  \:\!\bigr]$ consisting of a parity function matrix $C$ and a phase vector $\mathbf p$.
\end{defn}
Two distinct parity function tableaus $T = \bigl[\:\! C\,\big\vert\, \mathbf p  \:\!\bigr]$ and $T' = \bigl[\:\! C'\,\big\vert\, \mathbf p' \:\!\bigr]$ may represent the same state, if $T' = Q T$ for some $(N{+}1) \times (N{+}1)$ invertible matrix $Q$.
Such a transformation $Q$ represents a change of variables, in the summation expression of the state as described in \eqref{eqn:parityFormulaExpan}, leaving the overall sum invariant.
Note that such a matrix must satisfy $Q \mathbf e_0 = \mathbf e_0$: this corresponds to the fact that no change of variables can affect the value of constants.
Conversely, any invertible $(N{+}1) \times (N{+}1)$ matrix $Q$ which preserves the vector $\mathbf e_0 \in \mathbb Z_2^{N+1}$\!, may be used to transform a parity function tableau $T$ to an equivalent tableau (representing the same state) by left-multiplication.

In our application to QLNC circuits for a given qubit interaction network $G$, we may use an alternative representation, in which we write the qubit functions $f_j(a)$ next to the nodes corresponding to each qubit $j$ in the diagram of $G$.
For instance, the state illustrated in Fig.~\ref{f4} is the state $\ket{+}_1 \ket{+}_3 \ket{\mathrm{GHZ}_4}_{2,4,5,6}$ (with a phase function of zero).
This will prove practical when the objective is to demonstrate the effect of operations within a particular network $G$.

\subsubsection{QLNC operations on parity formula states}

We now consider how each of the transformations which are admitted in QLNC circuits may be simulated through transformations of parity function tableaus.

\paragraph{Simulating unitary gates.}

The effect of the unitary transformations CNOT, $X$, and $Z$ on parity formula states are easy to describe as transformations of their representations, by simply reasoning about the representation of the state as a superposition over the standard basis:
\begin{enumerate}[label=(\textit{\roman*})]
\item
  The effect of an $X$ operation on qubit $k$ is to update $f_k(a) \gets 1 + f_k(a)$;
\item
  The effect of a $Z$ operation on qubit $k$ is to update $\phi(a) \gets \phi(a) + f_k(a)$;
\item
  The effect of a CNOT operation with control $k$ and target $\ell$, is to update $f_\ell(a) \gets f_\ell(a) + f_k(a)$.
\end{enumerate}
It is easy to verify that these transformations correspond to elementary column transformations of the parity function tableau ${\bigl[\:\!C\:\!\big\vert\:\!\mathbf p\:\!\bigr]}$.
Specifically --- indexing the columns of $C$ from $0$ --- these operations may be realised respectively by adding the zeroth column of $C$ to the $k\textsuperscript{th}$ column, adding the $k\textsuperscript{th}$ column of $C$ to $\mathbf p$, and adding the $k\textsuperscript{th}$ column of $C$ to the $\ell\textsuperscript{th}$ column.
Note that these operations all preserve the rank of $C$.

\paragraph{Simulating projective measurements.}

The way in which we may represent measurements by transformations of a parity formula tableau is somewhat more complex, due to the possibility of state collapse.
To simplify the description of an $X$ or $Z$ observable measurement on a qubit $k$, we first perform a change of variables --- specifically, by putting the block matrix $T = {\bigl[\:\! C \:\!\big\vert\:\! \mathbf p \:\!\bigr]}$ in a reduced row-echelon form in which either column $k$ is a pivot column, or column $k$ is co-linear with $\mathbf e_0$ (so that $f_k(a)$ is a constant).
%This requires at most $O(n N^2)$ time to compute.
Suppose (without loss of generality) that ${\bigl[\:\! C \:\!\big\vert\:\! \mathbf p\:\!\bigr]}$ is already in such a reduced row echelon form, in which case either $f_k(a) = c_{k,0}$ is a constant function, or $f_k(a) = a_g$ for a single indeterminate indexed by $1 \le g \le N$; in the latter case, exactly one row of $C$ contains a $1$ in the $k\textsuperscript{th}$ column.
Having put the parity function tableau into reduced row-echelon form of this kind, we may then describe an $X$ or $Z$ observable measurement on qubit $k$, as follows.
\begin{enumerate}[label=(\textit{\roman*}), start=4]
\item
  For an $X$ measurement:
  \begin{enumerate}[label=({\alph*})]
  \item
    \label{item:XmeasProd}
    If $f_k(a) = a_g$ for an indeterminate which does not occur in any other qubit formula $f_j(a)$ --- \emph{i.e.},~if there is a single $1$ in the $g\textsuperscript{th}$ row of $C$ --- 
    then the state is unchanged by the measurement, and the measurement outcome is $s = (-1)^{p_g}$.
  \item
    \label{item:XmeasEnt}
    Otherwise, let $z_{N\!\!\:{+}\!\!\:1}$ be a new indeterminate (represented in $C$ by adding a new row at the bottom), and choose a measurement outcome $s = \pm 1$ uniformly at random.
    If $f_k(a)$ is constant prior to the measurement, then let $\Delta$ be the $(N{+}2)$-dimensional vector with a single $1$ in the final row; otherwise, let $\Delta$ be the $(N{+}2)$-dimensional column-vector with exactly two $1$s, in row $g$ and $N{+}1$ (counting from~$0$).
    We then add $\Delta$ to the $k\textsuperscript{th}$ column of $C$, and (in the case that $s = -1$) to $\mathbf p$ as well.
  \end{enumerate}
\end{enumerate}
  
\begin{itemize}
\item
  \textbf{Analysis of the state transformation.~~}
  In case~\ref{item:XmeasProd}, the state of qubit $k$ can be factored out of the sum, so that the state is either $\ket{+}$ (if $\phi$ lacks any $a_g$ term) or $\ket{-}$ (if $\phi$ contains a $z_g$ term), so that the measurement does not affect the state and the outcome is predetermined.
  Otherwise, in case~\ref{item:XmeasEnt}, qubit $k$ is maximally entangled with the rest of the system: the state has a Schmidt decomposition $\tfrac{1}{\sqrt 2} \ket{0}_k  \ket{A_0} + \tfrac{1}{\sqrt 2} \ket{1}_k \ket{A_1}$, where $\ket{A_b}$ in each case is the state-vector on the qubits apart from $k$ in the case that $a_g := b$ (possibly including a phase factor that depends on $a_g$).
  It follows that the outcome of the $X$ observable measurement is uniformly random, and that the state $\ket{A_\pm}$ of all of the other qubits will be in tensor product with $k$ after measurement.
  A straightforward calculation shows that $\ket{A_+} = \tfrac{1}{\sqrt 2} \sum_{x_g} \ket{\smash{A_{x_g}}}$ and $\ket{A_-} = \tfrac{1}{\sqrt 2} \sum_{x_g} (-1)^{x_g} \ket{\smash{A_{x_g}}}$; these are the states described by simply omitting the $k\textsuperscript{th}$ column of $C$, and (in the case of $\ket{A_-}$) adding an extra $a_g$ term to the phase function.
  To represent the post-measurement state, it suffices to introduce a new indeterminate $a_{N+1}$ to represent the independent superposition on qubit $k$; for the post-measurement state $\ket{-}_k$, we also must add $a_{N+1}$ to the phase function.

\item
  \textbf{On the rank of the resulting parity function matrix.~~}
  Note, above, that in case~\ref{item:XmeasProd} there is no change in the tableau, and thus no change in the rank of $C$.
  In case~\ref{item:XmeasEnt}, we must consider two sub-cases: one where $f_k(a) = c_{k,0}$ before the measurement, and one where $f_k(a) = a_{g}$ before the measurement.
  In either case, we add one row, in which the only non-zero entry is in column $k$.
  In the former case, we add one row and add a coefficient $1$ in column $k$ in that bottom row.
  This increases both the number of rows and the rank.
  In the latter case, we consider the operations performed at column $k$ in two steps: first setting the coefficient at row $g$ to zero, then setting the coefficient in the new row $N+1$ to one.
  Setting the coefficient at row $g$ to zero does not decrease the rank: the column $k$ cannot any longer be a pivot column.
  Prior to the first step, the $k\textsuperscript{th}$ column is a pivot column; but we may alternatively select any other column in which the $g\textsuperscript{th}$ row is set to $1$, as (by construction) these columns do not contain a pivot position for any other row.
  Thus, setting the $g\textsuperscript{th}$ coefficient of the $k\textsuperscript{th}$ row does not decrease the rank; and again, adding a row in which only the $k\textsuperscript{th}$ column has a $1$ increases both the rank and the number of columns.
  Thus, this operation maintains the property of $C$ having a rank equal to the number of its rows.
\end{itemize}

\begin{enumerate}[label=(\textit{\roman*}), start=5]
\item
  For a $Z$ measurement:
  \begin{enumerate}[label=({\alph*})]
  \item
    \label{item:ZmeasProd}
    If $f_k(a) = c$ is a constant function, then the measurement leaves the state unchanged, and the measurement outcome is $(-1)^c$.
  \item
    \label{item:ZmeasEnt}
    Otherwise, we select a measurement outcome $s = (-1)^b$ for a bit $b \in \{0,1\}$ chosen uniformly at random.
    Let $\Delta = b \mathbf e_0 + \mathbf c_k$.
    Add $\Delta$ to all columns of $T = {\bigl[\:\!C\:\!\big\vert\:\!\mathbf p\:\!\bigr]}$ which contain a $1$ in the $g\textsuperscript{th}$ row (including the $k\textsuperscript{th}$ column itself), and remove the row $g$ entirely from the tableau.
  \end{enumerate}
\end{enumerate}
  
\begin{itemize}
\item
  \textbf{Analysis of the state transformation.~~}
  In case~\ref{item:ZmeasProd}, it is obvious that qubit $k$ is in a fixed state: the outcome will be $+1$ if it is in the state $\ket{0}$, and $-1$ if it is in the state $\ket{1}$.
  Otherwise, in case~\ref{item:ZmeasEnt}, the state of the system can again be described as a superposition $\tfrac{1}{\sqrt 2} \ket{0}_k  \ket{A_0} + \tfrac{1}{\sqrt 2} \ket{1}_k \ket{A_1}$, albeit where it is possible in principle that $\ket{A_0} = \pm  \ket{A_1}$.
  We may simulate the assignment of the $k\textsuperscript{th}$ qubit to $b$ by quotienting out all of the functions $f_j(a)$ and the phase function $\phi(a)$ by the relation $a_g + b = 0$.
  We may do this in effect by adding the column vector $\Delta$ defined above to all columns with a non-zero coefficient in the row $g$, thereby obtaining a tableau in which the $g\textsuperscript{th}$ row is empty.
  This corresponds to a state in which the variable $a_g$ no longer plays any role; together with the updated normalisation after measurement, we may represent this by removing row $g$.
    
\item
  \textbf{On the rank of the resulting parity function matrix.~~}
  Note, above, that in case~\ref{item:ZmeasProd} there is no change in the tableau, and thus no change in the rank of $C$.
  In case~\ref{item:ZmeasEnt}, we may without loss of generality suppose that the $k\textsuperscript{th}$ column is the last column to which we add $\Delta$.
  In each case, the vector is added to a non-pivot column, in which case this cannot decrease the rank; nor will it increase the rank, as it only sets coefficients to $0$ in rows which already have pivot positions.
  These column additions preserve the property of being a reduced row-echelon form.
  The final addition of $\Delta$ does decrease the rank by $1$, as it turns the $g\textsuperscript{th}$ row from a non-zero row-vector (in a reduced echelon form) to a zero row.
  Thus the rank of the parity function matrix $C$ decreases by $1$; as removing row $g$ from the tableau reduces the number of columns by $1$, this operation maintains the property of $C$ having a rank equal to the number of its rows.
\end{itemize}
From the above, we see that we may represent QLNC operations by simple transformations of a parity function tableau $T = {\bigl[\:\!C\:\!\big\vert\:\!\mathbf p\:\!\bigr]}$, which in particular preserves an invariant that the rank of the parity function matrix $C$ is equal to the number of its rows.

\paragraph{Simulating destructive measurements and qubit preparations.}

One might reasonably wish to regard some measurements as being \emph{destructive}, \emph{i.e.},~in not leaving any post-measurement state.
We may simulate this by simply removing from $C$ the column corresponding to the destructively measured qubit, and removing from the entire tableau any row for which (after the column removal) the matrix $C$ is entirely zero.
Conversely, one may simulate the preparation of a fresh qubit in a standard basis state $\ket{b}$ for $b \in \{0,1\}$, by inserting a new column into $C$ with the value $b \mathbf e_0$.
To instead simulate the introduction of a fresh qubit in the state $\tfrac{1}{\sqrt 2} \sum_{x'} (-1)^{bx'} \ket{x'}$ for $b \in \{0,1\}$, one may insert a new row into the tableau (at the bottom of both $C$ and $\mathbf p$) which is entirely zero, then setting the new coefficient of $\mathbf p$ in this row to $b$ if this is different from $0$, and then inserting a new column into $C$ which has only a single $1$ in the final row.

\vspace*{-1ex}
\paragraph{Terminology.} For the sake of definiteness, ``the QLNC formalism'' will refer below to the techniques described above to describe transformations of parity function tableaus (or some some equivalent representation), as a means to simulate stabiliser circuits composed of this limited set of operations.

\subsubsection{Depicting and simulating QLNC circuits}
\label{sim}

Having defined the QLNC formalism, we now demonstrate how it may be used to simulate QLNC circuits.
In this context, we will prefer to represent the parity function states diagrammatically rather than as a matrix --- and to represent it together with a visual representation of the transformations to be performed.
\begin{enumerate}
\item
    Each vertex is a qubit $j$, with an associated formula $f_j(a)$, for some symbols $a_1, \ldots, a_N$.
    The initial formulae for each qubit is generally very simple: each qubit prepared in the $\ket{+}$ state is assigned the formula $f_j(a) = a_j$ for a unique formal symbol $a_j$, and each qubit initialised in the $\ket{0}$ state is assigned the formula $f_j(a) = 0$.
\item
    Pauli $X$ gates on a qubit $k$, which are classically conditioned by the outcome of a $Z$-observable measurement of a different qubit $j$, are represented as doubled lines with an orientation from $j$ to $k$.
    Coherently controlled CNOT gates are drawn along edges of the network $G$.
\item
    One or more qubits may also be simultaneously ``terminated'', in which case they are measured with the $X$ observable.
    The outcome may then be used to control Pauli $Z$ operations to cancel out the relative phases which arise as a result of any $-1$ measurement outcomes.
\item
    There is a time-ordering of the operations represented by the edges are performed.
    In simple QLNC circuits, this is represented by a single integer at each edge, and an integer inside each node to be terminated.
    (Two edges which meet at a common vertex, and which are not both classically controlled $X$ gates, must be performed at different times, and thus must be assigned different numbers. Also, no edge can have the same number as the termination of a node to which it is incident. Otherwise, there are no constraints.)
    More generally, it will be reasonable to consider QLNC circuits in which edges are used some constant number of times, \emph{e.g.}~up to two times; we would then label edges by a list (or set) of those times in which it is used, and the operations involving a common vertex must be disjoint (again, unless those operations are all classically controlled $X$ gates).
\end{enumerate}
\paragraph{Remarks on termination.}
It may not be immediately obvious that the claim made about termination --- that any relative phases induced by obtaining measurement outcomes of $-1$ from $X$ observable measurements --- can be ``undone'', leaving a state which is a uniform superposition over some set of standard basis states (\emph{i.e.},~with no relative phases at all).
In the case of a QLNC circuit which (successfully) simulates a classical linear network code, this may be more plausible to the reader.
In fact, we make a stronger claim:
\begin{lemma}
    For any state $\ket{\psi}$ given by a parity function tableau $T = {\bigl[\:\! C \:\!\big\vert\!\: \mathbf p \!\:\bigr]}$ on $n$ qubits, it is possible (in time dominated by Gaussian elimination on $T$) to find a subset $S \subset \{1,\ldots,n\}$ of qubits, such that $\ket{\psi'} = Z^{\otimes S} \ket{\psi}$ has a parity function tableau $T = {\bigl[\:\! C \:\!\big\vert\!\: \mathbf 0 \!\:\bigr]}$.
\end{lemma}
\noindent
We prove this result here, to demonstrate that ``termination'' is a well-defined operation in principle.
\begin{proof}
    Let $Q$ be an invertible linear transformation for which $Q T = \bigl[\mathbf e_0 \:\: \tilde{\mathbf c}_1 \:\: \tilde{\mathbf c}_2 \:\: \cdots \:\: \tilde{\mathbf c}_n \:\: \tilde{\mathbf p}\bigr]$ is in reduced row-echelon form, and let $\tilde f_j$ be the qubit function corresponding to column $\tilde{\mathbf c}_j$.
    Then, for every formal indeterminate $z_g$, there is a qubit $k_g \in \{1,\ldots,n\}$ for which $f_{k_g} = z_g$.
    Let $J$ be the set of rows for which $\tilde p_j = 1$, and let $S = \{ k_j \mid j \in J \}$.
    Then the effect of $Z^{\otimes S}$ is to map $\tilde{\mathbf p} \mapsto \mathbf 0$.
    This may be represented by a transformation $R$ for which $QTR = \bigl[\mathbf e_0 \:\: \tilde{\mathbf c}_1 \:\: \tilde{\mathbf c}_2 \:\: \cdots \:\: \tilde{\mathbf c}_n \:\: {\mathbf 0}\bigr]$, which is a parity function tableau for a state without relative phases over the standard basis.
    (Indeed, it follows that the final column of $TR$ is also $\mathbf 0$, so that simulating $Z^{\otimes S}$ on the original tableau removes all relative phases without committing to the change of variables described by $Q$.)
\end{proof}
As a corollary, it follows that for a parity function state, we can induce whatever relative phase we like, of the form $(-1)^{\phi(x)}$ for any linear function $\phi$ of the indeterminates.
We may use this to justify the notion of ``terminating'' one qubit independently of any others, and ``undoing'' any change to the phase function which occurs as a result of obtaining a $-1$ outcome.
The specific choice of qubits on which to perform the $Z$ operations may not be unique, but it suffices for our purposes that such a set can always be found efficiently.

The way one might use the QLNC formalism to simulate a particular QLNC circuit is illustrated in Fig.~\ref{f4}.
This example distributes two Bell states across a rectangular grid, by simulating the classical Butterfly network protocol with some ``out-of-order'' evaluations.
To compensate for the out-of-order evaluation, classically controlled $X$ operations are required upon the measurement of one of the qubits: this is in effect a coding operation using a link outside of $G$, relying on our architectural assumption that classical information can be communicated more freely than quantum information.

\begin{figure}[!t]
	\centering
	\includegraphics[width=0.88\linewidth]{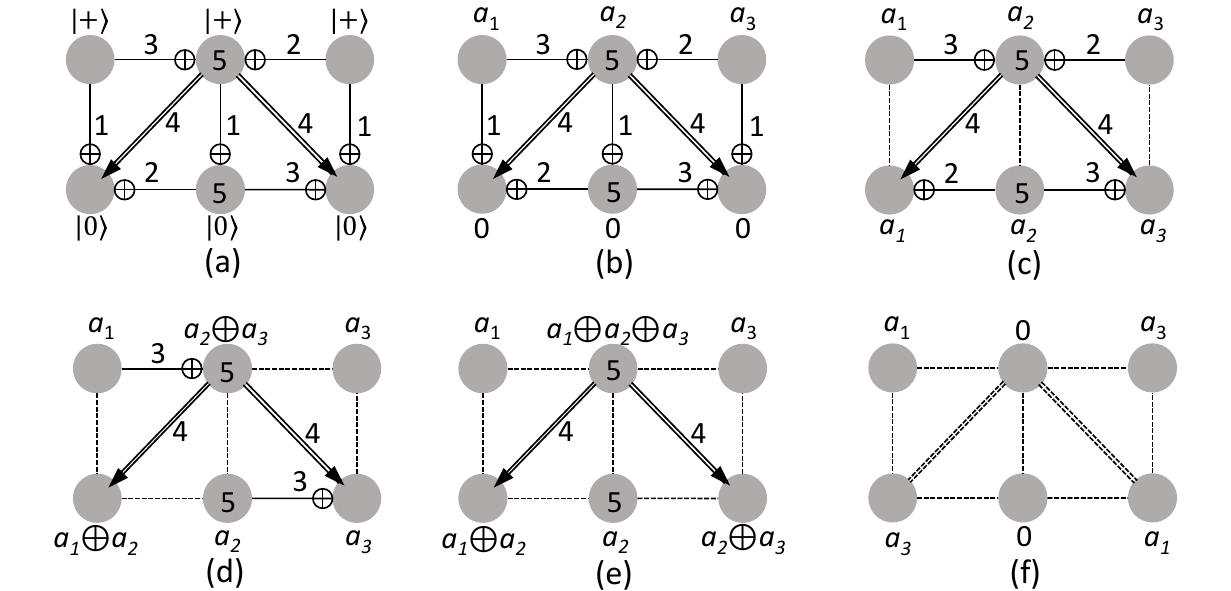}
	\captionsetup{width=0.95\linewidth}
	\caption{\small{Example of the out of order Butterfly: (a) the order of edges, slightly different, but with equivalent quantum circuit to that given in Fig.~\ref{f1}; (b) the initial labelling of the qubits; (c) the labels after edges ``1''; (d) the labels after edges ``2''; (e) the labels after edges ``3''; (f) the labels after the classical control (edges ``4'') and the terminations (the fifth layer of operations, denoted by the nodes labelled ``5'').}}
	\label{f4}
\end{figure}

\subsection{Using the QLNC formalism to design entanglement distribution circuits}
\label{example}

As already noted, the purpose of developing the QLNC formalism is to enable the use of classical linear network codes as a basis to design entanglement distribution quantum circuits. We begin by noting that there are situations in which QLNC circuits can distribute entanglement which do not correspond to linear network codes.

\subsubsection{Shallow QLNC circuits for entanglement distribution}

The classical-quantum linear network code result suggests a number of ways in which QLNC circuits can be used to distribute entanglement. In this section we detail one such application, prompted by our desire to use classical linear network coding results and intuitions to distribute entanglement \textit{efficiently} (i.e., with as few layers of quantum operations as possible).

We consider the following scenario.
Let there be a classical binary linear network code over a network, connecting $k$ transmitter-receiver pairs; and let that network consist of three types of nodes:
\begin{itemize}
    \item Transmitter nodes --- for which each incident edge is outbound (i.e., a directed edge with direction away from the transmitter), and the transmitter broadcasts its bitsteam on all incident edges.
    \item Relay nodes --- that have an arbitrary number of input and output edges, and whose operation is to broadcast the modulo-2 sum of all incoming bitstreams on all of the output edges.
    \item Receiver nodes --- for which each incident edge is inbound, and whose operation is to perform the modulo-2 sum of all of the incoming bitstreams, which yields the desired bitsream (i.e., that transmitted by the corresponding paired transmitter).
\end{itemize}
With the three types of nodes (graph vertices) defined thusly, we can prove an important result about the required depth of layers of quantum operations, when the qubit interaction graph is again $G=\{V,E\}$.
\begin{thm}
\label{thm:constdepth}
If a multiple multicast (including multiple unicast as a special case) classical binary linear network code exists over a network $G$, from a set of transmitter nodes $T= \{t_1 , \cdots , t_N\}$ with in-degree $0$ to a corresponding set of receiver nodes $R_j = \{r_{j,1} , \cdots , r_{j,n_j} \}$ with out-degree $0$, then there is a QLNC circuit whose \textup{CNOT} operations are located along the edges of $G$ and distributes $\ket{\Phi^+}$ and $\ket{\mathrm{GHZ}}$ states between corresponding transmitter/receiver node sets.
Moreover, this circuit has quantum depth $\le {2 (\chi{-}1) (\delta{+}2) + 1}$, where $\delta$ is the largest in/out degree of any vertex in $G$, and $\chi$ is the chromatic number of $G$.
\end{thm} 
\begin{remark}
It is in general NP-complete to compute the vertex-chromatic number $\chi$ of a network.
However, in many realistic cases it will be easy to compute $\chi$.
For instance, bipartite networks (such as tilings by squares or hexagons, or connected subgraphs of these) have $\chi = 2$ by definition.
In any more complex network $G$, we may alternatively substitute $\chi$ with the number of colours of any proper vertex--colouring that one may find.
For instance, any graph has a $\deg(G)+1$ vertex-colouring which can be found in polynomial time~\cite{greedy}.
Furthermore, we have $\chi \le 4$ in planar architectures (\emph{i.e.},~for which $G$ is a planar graph) by the Four Colour Theorem~\cite{fourcolour}, and an explicit four-colouring can also be found in polynomial time~\cite{fourcolourp}.
% In the proof of Theorem~\ref{thm:constdepth} below, we suppose that some proper vertex-colouring is provided.
\end{remark}
\begin{remark}
  The role of $\delta$ in Theorem~\ref{thm:constdepth} is in relation to the number of colours of an edge-colouring $\gamma$ of $G$, such that no two edges with the same colour label leave a common vertex or enter a common vertex.
  (We call such a colouring a ``proper directed-edge colouring''.)
  If we transform $G$ into a graph $\tilde G$, in which each vertex $q$ is replaced with a vertex $q_i$ (inheriting only the in-bound edges of $q$) and a vertex $q_o$ (inheriting only the out-bound edges of $q$), then $\delta$ is the maximum degree of $\tilde G$, and $\gamma$ corresponds to a proper edge-colouring of $\tilde G$.
  By Vizing's theorem \cite{Vizing1964}, the edge-chromatic number of $\tilde G$ is at most $\delta + 1$, and an edge-colouring with $\delta + 1$ colours can be found efficiently.
  (An edge-colouring of $\tilde G$ must have at least $\delta$ colours; and it may be easy to find an edge-colouring of this kind, \emph{e.g.},~if $G$ arises from a lattice in the plane.
  If one may find such a colouring, the bound above improves to ${2 (\chi{-}1) (\delta{+}1) + 1}$.
  For the square lattice, with $\chi = 2$ and with $\delta=3$ if no vertex has four in-edges, this yields the bound of $9$ described on page~\pageref{discn:introSquareLatticeBound}.)
\end{remark}

\begin{proof}
Let $c: (V \cup E) \to \mathbb N$ be a colouring of the nodes and edges, such that $c$ provides a proper colouring $1, 2, \ldots, A$ to the nodes of $G$, and also a proper directed-edge colouring $1, 2, \ldots, B \le \delta+1$ to the edges of $G$.
Consider the following procedure:
\begin{enumerate}
\item
  Initialise all of the qubits, where each qubit $q$ is initialised either in the state $\ket{0}$ if it has no outgoing edges, or if it has some neighbour $p$ by an in-bound edge (that is, there is a directed edge from $p$ to $q$) for which $c(p) < c(q)$, and is initialised in the state $\ket{+}$ otherwise.
  (In the QLNC formalism, we associate a formal indeterminate $a_q$ with each qubit $q$ initialised in the $\ket{+}$ state.)
\item
  For each non-receiver node $q$ with $c(q) = 1$ in parallel, perform the following procedure:
  \begin{itemize}
  \item
    For each $1 \le j \le B$, perform a CNOT operation on any edge $e$ with $c(e) = j$ leaving $q$.
    (In the QLNC formalism, this adds $a_q$ to the formula $f_v(a)$ for the node $v$ at the end of $e$.) 
  \end{itemize}
\item
  For each $2 \le h \le A \!-\! 1$, perform the following operations in parallel on non-receiver nodes $q$ with $c(q) = h$:
  \begin{enumerate}[label=\alph*.]
  \item
    For each $1 \le j \le B$, perform a CNOT operation on any edge $e$ with $c(e) = j$ leaving $q$.
%    (This adds the current value of $f_q(a)$ to the formula $f_v(a)$ for the node $v$ at the end of $e$.) 
  \item
    If $f_q(a) \ne a_q$ (\emph{i.e.},~$q$ was a target of some CNOT or Pauli X operation before this round):
    \begin{enumerate}[label=(\textit{\roman*})]
    \item
      Terminate the qubit $q$, by performing an $X$ observable measurement.
    \item
      If the outcome of the preceding measurement is $-1$, perform $Z$ operations on an appropriate set of qubits, and a $Z$ operation on $q$ to transform the post-measurement state of $q$ from $\ket{-}$ to $\ket{+}$.
      (If any qubit $v$ has been selected to be subject to a $Z$ operation by multiple qubits $q$ with $c(q) = h$, we perform $Z$ on $v$ if and only if the number of such qubits $q$ is odd.)
    \item
      If $q$ has any neighbours $p$ by in-bound edges, such that $c(p) > c(q)$, then for each $1 \le j \le b$, perform any CNOT operations on edges $e$ with $c(e) = j$, which are outgoing from $q$.
      (In the QLNC formalism, this adds $a_q$ to the node-formula $f_v(a)$ for the node $v$ at the end of $e$.)
    \end{enumerate}
  \end{enumerate}
  \item
    For each non-receiver node $q$ with $c(q) = A$ in parallel, perform the following procedure:
    \begin{itemize}
    \item 
        For each $1 \le j \le B$, perform a CNOT operation on any edge $e$ with $c(e) = j$ leaving $q$.
%        (Again, this adds the current value of $f_q(a)$ to the formula $f_v(a)$ for the node $v$ at the end of $e$.)
    \end{itemize}
  \item
    For each relay qubit $q$: if $c(q) < A$, perform a $Z$-observable measurement on $q$, obtaining an outcome $y_q \in \{0,1\}$.
    Otherwise, terminate $q$ (\emph{i.e.}, measure $q$ with the $X$ observable and perform appropriate $Z$ corrections).  
  \item
      For receiver nodes $q$ and for the relay qubits $q$ for which $c(q) < A$, recursively define the \emph{delayed signal correction} $w_q$ as the sum mod $2$ of $(y_r \oplus w_r)$, for $r$ ranging over all neighbours of $q$ via incoming links for which $c(r) < A$ (or zero, if there are no such nodes $r$).%
        \footnote{%
            This recursive definition terminates in the relay nodes with no such incoming links.
            Note that the outcomes $z_r$ on which $w_q$ depends may be determined in advance for any receiver node $q$; from this the value of $w_q$ may be computed in logarithmic depth from that formula using standard parity-computation techniques.
        }
    Defining $w_q$ in this way, perform a Pauli-$X$ gate on each receiver node $q$ for which $w_q = 1$.
    (If such a qubit is subject to a $Z$ correction as a part of Step~5, we perform a $Y$ operation instead of the $X$ and $Z$ operations.)
\end{enumerate}
We compute the quantum depth of this procedure as follows.
The operations of Step~1 has depth $1$.
Both steps~2 and~4 have depth at most $B$.
Step~3 is a loop with $A {\!\:-\!\:} 2$ iterations, in which part~(a) has depth at most $B$, and part~(b) has depth at most $B {\!\:+\!\:} 2$.
Steps~5 and~6 together have depth at most $2$.
Together, the depth is then $1 + B + (A\!\:{-}\!\:2)(2B\!\:{+}\!\:2) + B + 2 = 2(A\!\:{-}\!\:1)(B\!\:{+}\!\:1) + 1 \le 2(\chi\!\:{-}\!\:1)(\delta\!\:{+}\!\:2) + 1$.

Fig.~\ref{appf1} shows a sketch of why this procedure works.
In effect, we wish for ``information'' (more precisely: correlation of values of a qubit in the standard basis, when taken in superposition) to be transmitted through each relay node, from each of its sources (summing the signals from these sources modulo 2) to the qubits on each of its outward links.
Some of this information may accumulate at a given relay node $q$ before round $c(q)$, in which case it is explicitly passed on through a CNOT.
The rest accumulates at $q$ after round $c(q)$, and also after the node $c(q)$ has communicated a formal indeterminate $a_q$ on each of its outgoing links.
If we may collapse the state in such a way to assign to $a_q$ the modulo~2 sum of the remaining signals from its incoming links (accumulated after round $c(q)$), this collapse will complete the transmission of the information from the inbound links of $q$ through to the outbound links.

\begin{figure}[!t]
	\centering
	\includegraphics[width=0.73\linewidth]{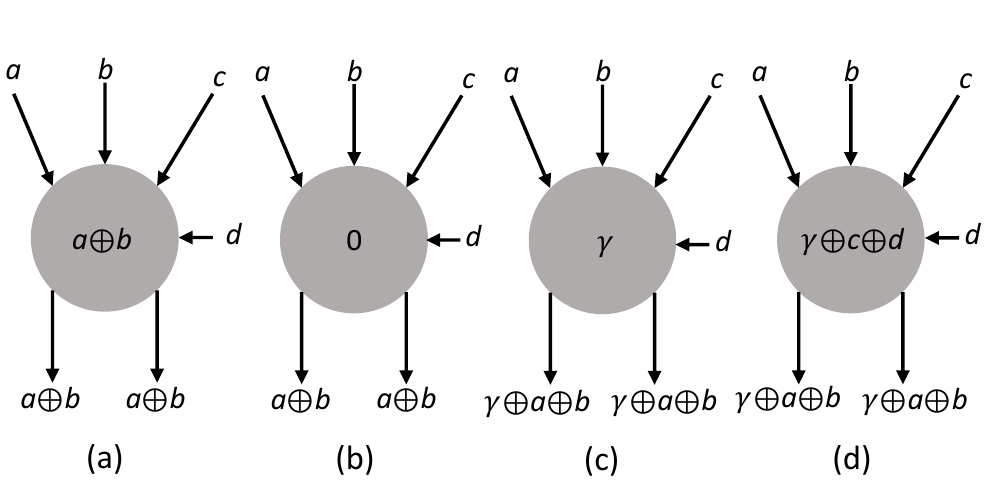}
	\captionsetup{width=0.95\linewidth}
	\caption{\small{Example of a relay node, whose operation in the linear network code is to forward $a \oplus b \oplus c \oplus d$ to all three outgoing edges. If the vertex colouring is such that the turn of this vertex is after incoming symbols $a$ and $b$ have arrived, but before $c$ and $d$ have, then the procedure continues as follows: in (a) $a \oplus b$ (i.e., the current label of the vertex) is forwarded to all outgoing edges; (b) the qubit is terminated and set to zero; (c) the qubit is set to the $\ket{+}$ state, and given the new label $\gamma$, which is then forwarded to all of the outgoing edges, therefore meaning that over the two rounds of forwarding $a \oplus b \oplus \gamma$ has been forwarded; the qubit then waits for the remainder of the process to complete, after which all edges will have been performed, so its label will now be $c \oplus d \oplus \gamma$, which can then be measured and corrected such that $c\oplus d = \gamma$, which then means that $a \oplus b \oplus c \oplus d$ has been forwarded as required.}}
	\label{appf1}
\end{figure}

More formally, consider the node formulae which result from this procedure.
\begin{itemize}
\item 
    For each relay node, let $Z_p(a)$ denote the boolean formula which is transmitted to it on an incoming link from a node $p$ for which $c(p) < c(q)$.
    We will then have $Z_p(a) = a_p \oplus E_p(a)$, where $E_p(a)$ is the modulo~2 sum of the corresponding functions $Z_r(a)$ for nodes with edges towards $p$ such that $c(r) < c(p)$.
\item
    The formula which is stored at qubit $p$ just prior to its measurement in Step~4 is the formula $Y_p(a) = a_p + L_p(a)$, where $L_p(a)$ is the modulo~2 sum of $Y_r(a)$ for nodes $r$ with links inwards to $p$ such that $c(r) > c(p)$.
    (An outcome $y_p$ represents the value of $Y_p(a)$ which is produced by a standard basis measurement of $p$.)
\end{itemize}
If in Step~4 we measure qubit $p$ and collapse it to the state $\ket{0}$, we in effect condition on the situation that $a_p = L_p(a)$.
(In the event that we obtain $\ket{1}$, we perform corrections which allow us to simulate having obtained $\ket{0}$ instead.)
This produces an acyclic graph of formula substitutions, from the node-formulae of the transmitters to the node-formulae of the receivers.
By induction on the substitution depth (\emph{i.e.},~the distance of relay nodes from any receiver node), we may show that performing the necessary substitutions in the formula for $Z_p(a)$ yields the information which, in the classical linear protocol, $p$ would transmit on its outgoing links.
It follows that the parity function computed at each receiver node is the function $a_t$ (for its corresponding transmitter node $t$) that is computed in the classical linear network code.
\end{proof}

In the protocol above, each relay is measured twice (\emph{i.e.},~for the termination, and then at the end to resolve the extra formal label introduced).
For this reason, it is necessary to strictly separate transmitters, receivers and relays.
However, this setting is not too restrictive, and corresponds to examples of classical linear network codes such as we see in Figs.~\ref{f01} and~\ref{f01a}.

Note that while Steps 2, 3a, 3b(\textit{iii}), and 4 of our protocol iterate through all edge colours $1 \leq j \leq B$, the only edge-colours that contribute to the depth are those associated to edges which leave some \emph{vertex} of the colour $1 \le h \le A$ being considered in the given step.
Thus the bound above will often be loose, and in fact it may be possible to find better realisations using a suitably tuned directed edge-colouring of $G$.
However, our result obtains regardless which edge-colouring one uses, so long as it uses at most $\delta+1$ edge-colours, which again may easily be found~\cite{Vizing1964}.

\subsubsection{Example of QLNC solution involving entanglement swapping, for which no classical linear network coding solution exists}
\label{eswap}

\begin{figure}[!t]
	\centering
	\includegraphics[width=\linewidth]{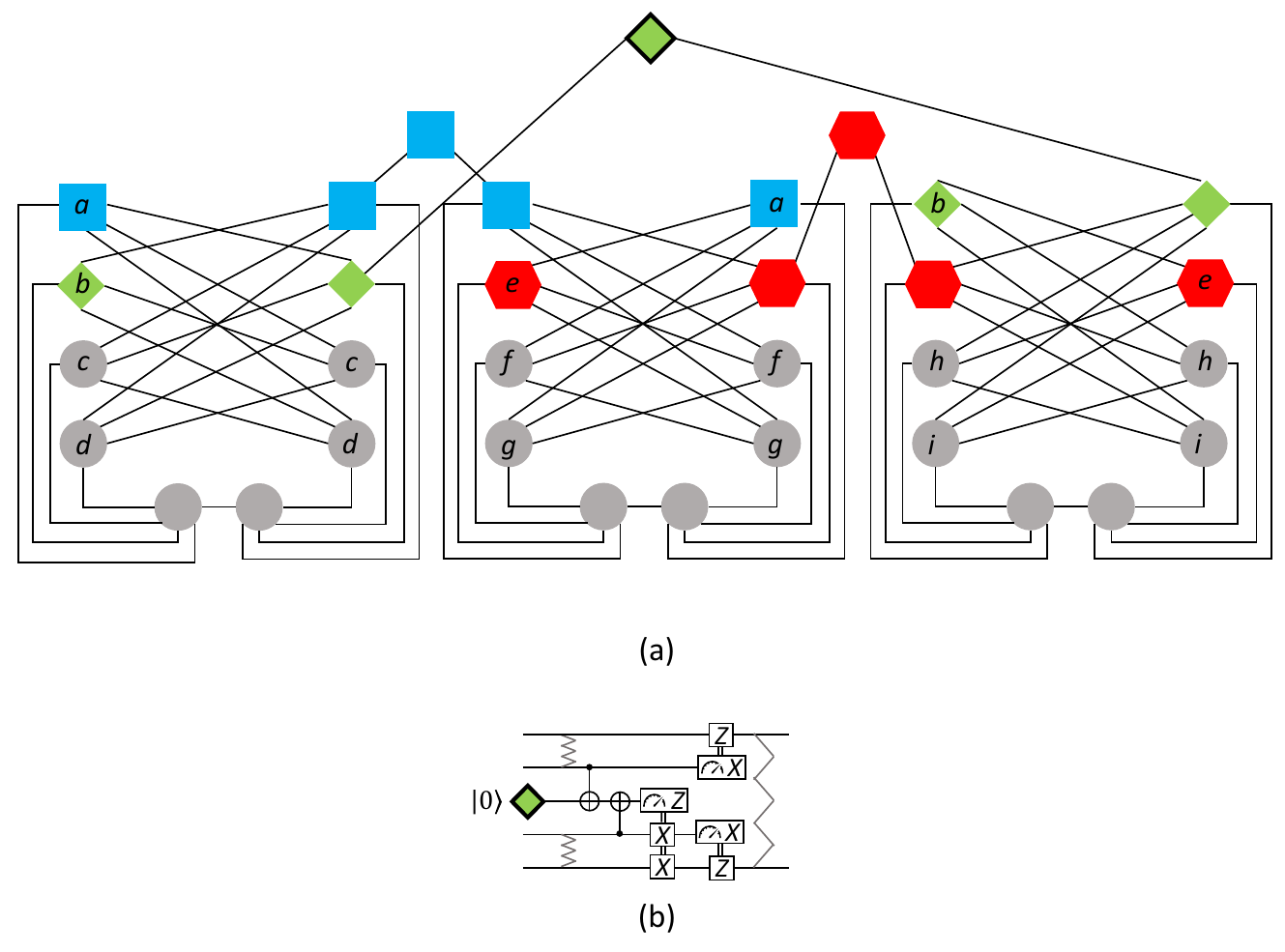}
	\captionsetup{width=0.95\linewidth}
	\caption{\small{An example of a composite network with a QLNC circuit but no (classical) linear network code. Note that this composite network corresponds to three connected copies of the network in Fig.~\ref{f025}, here we draw the part of the graph in Fig.~\ref{f025}(b) as the two bottom-most nodes of each component.}}
	\label{appf2}
\end{figure}

When using linear network codes to design QLNCs, in which we distribute entangled pairs, we are free to assign which half of each desired Bell state corresponds to the transmitter in the linear network code and which half represents the receiver. However, while we have the freedom to decide which is the transmitter and which is the receiver, it may be the case that deciding the transmitter / receiver assignment for one Bell state fixes that for the rest. For example, if we consider the corresponding QLNC to the linear network code shown in Fig.~\ref{f025}, we can see that, even if we allow the links to be bi-directional, we must have one column of receivers and one column of receivers. That is, we cannot find a linear network code for the case where some of the left-hand nodes are receivers and some are transmitters.

This principle allows us to construct composite networks, in which some data must flow through multiple networks such that there is no linear network code. This is the case shown in Fig.~\ref{appf2}(a), composed of three copies of the network shown in Fig.~\ref{f025} with some extra links, in which each pair of letters is to be connected. Even if we are free to assign which of each pair of letters is the transmitter and the receiver, and also the direction of each link, we still cannot find a linear network code. This can be seen by considering the non-circular coloured nodes, which correspond to data which must flow through two of the component networks. Using the linear network code of Fig.~\ref{f025}, we can connect the pairs ``$cc$'' and ``$dd$'', as well as propagating the left-hand occurrences of ``$a$'' and ``$b$'' forward from the left-hand column of vertices to the second from left. The left hand occurrence of $a$ can now be forwarded via the intermediate blue square node, and the same left-to-right linear network code can be performed on the middle of the three component graphs, and then again forwarding ``$e$'' through the intermediate red hexagonal node to the right-hand of the three component graphs. Once again, we perform the same left-to-right linear network code on the right-hand of the three graphs, which means that we have now connected all of the pairs of letters, with the exception of $b$. In the case of $b$, each of the two $b$s has been forwarded as if it were at a transmitter, and they are connected by a common receiver --- the top-most node, which is a green diamond with a thick black ring.

Obviously, this is not a linear network code, as we have not connected pairs of letters as if one were a transmitter and the other a receiver (i.e., by a continuous data-flow from transmitter to receiver), however we \textit{can} find a QLNC circuit, as routing each letter towards a common receiver (the black-ringed green diamond node) can yield the desired Bell state by entanglement swapping in the black-ringed node, as shown in Fig.~\ref{appf2}(b).

A similar argument can be made for there not being a linear network code even if the component linear network codes are run right-to-left, in which case the black-ringed node would look like a single transmitter forwarding data to two receivers (the nodes marked b). A situation which can also be implemented as a QLNC circuit (\emph{i.e.}, if the black ringed node is terminated at the end) that does not correspond to any linear network code with the transmitter-receiver pairs as designated by the symbols in Fig.~\ref{appf2}.

\subsubsection{Classical-quantum linear network codes}

In Section~\ref{outoforder} we saw how the network codes in QLNC circuits can be performed ``out of order'' in some cases, and in Section~\ref{eswap} we gave an example of the use of entanglement swapping to implement a linear network code as if two transmitters are routing towards a common receiver. These are two instances of a general principle we notice that \emph{together the} CNOT \emph{operations and classical control must form a linear network code}. That is, if we consider the following situation:
\begin{enumerate}
    \item We have $n$ qubits connected in a network $G = \{V,E\}$, in which each edge means that a single CNOT (in either direction) is allowed.
    \item We allow classically-controlled $X$ gates (conditioned on measurement outcomes). It is convenient to consider this possibility of a classical control conditioned on the measurement outcome of a different vertex as a completely connected graph $K_n$ on the same set of vertices (where $n = \lvert V(G) \rvert$). That is, each edge represents the possibility of performing a classically controlled Pauli-$X$ gate.
\end{enumerate}
These coherently- and classically-controlled operations represent two of four primitives that we allow, the others being:
\begin{enumerate}
\setcounter{enumi}{2}
    \item Initialisation of qubits in the $\ket{+}$ or $\ket{0}$ state.
    \item Termination of the qubits according to the process described in the Section~\ref{sim}.
\end{enumerate}

\begin{thm}
\label{mainthm1}
Consider a multiple multicast (including multiple unicast as a special case) classical binary linear network code exists over any subgraph of the graph $G'= G \cup K_n$ sending a unit rate bitstream, where each edge of the graph is a unit rate bi-directional edge (but not allowing fractional routing).
Suppose that this code has a set of transmitting source vertices $T' = \{t_1, \ldots, t_{N'}\}$ for some $N' > 0$, where the first $N < N'$ of these have corresponding recevier sets $R_j = \{r_{j,1} , \cdots , r_{j,n_j} \}$ for $1 \le j \le N$ (with the remaining transmitters $t_{N{+}1}, \ldots, t_{N'}$ having signals which are not necessarily actually received by any ``receiver'' nodes).
Suppose further that \textup{(a)}~the information transmitted on the edges of $K_n$ from any single node, and \textup{(b)}~the information transmitted by the nodes $t_1$ through $t_N$, are linearly independent of each other.
Then by simulating this linear network code by a QLNC circuit, with \textup{CNOT} operations restricted to the same graph $G$ and classically-controlled Pauli operations oriented along the other edges, the resulting protocol generates a product of $\ket{\Phi^+}$ and $\ket{\mathrm{GHZ}}$ states, where each $\ket{\Phi^+}$ or $\ket{\mathrm{GHZ}}$ is over each of the sets $\{ t_j, r_{j,1} , \cdots , r_{j,n_j} \}$ for all $1 \le j \le N$.
\end{thm}
\begin{proof}[Proof (sketch)]
The core of the proof is showing that the QLNC formalism described correctly keeps track of the quantum state, which follows from the formalism description in Section~\ref{main}. We provide an explicit proof of the Theorem in Section~\ref{app1}, which explains why general QLNC circuits of this form achieve the desired result, and also serves to give a detailed walk through illustrating precisely how the QLNC formalism (including terminations) correctly simulates QLNC circuits in practise.
\end{proof}
An important special case occurs when the linear network code only requires edges in the graph $G$.
\begin{corollary}
\label{corr}
Consider a multiple multicast (including multiple unicast as a special case) classical binary linear network code exists over any subgraph of the graph $G$ sending a unit rate bitstream, where each edge of the graph is a unit rate bi-directional edge (but not allowing fractional routing).
Suppose that this code has a set of transmitting source vertices $T' = \{t_1, \ldots, t_{N'}\}$ for some $N' > 0$, where the first $N < N'$ of these have corresponding receiver sets $R_j = \{r_{j,1} , \cdots , r_{j,n_j} \}$ for $1 \le j \le N$ (with the remaining transmitters $t_{N{+}1}, \ldots, t_{N'}$ having signals which are not necessarily actually received by any ``receiver'' nodes).
Then by simulating this linear network code by a QLNC circuit, with \textup{CNOT} operations restricted to the same graph $G$, the resulting protocol generates a product of $\ket{\Phi^+}$ and $\ket{\mathrm{GHZ}}$ states, where each $\ket{\Phi^+}$ or $\ket{\mathrm{GHZ}}$ is over each of the sets $\{ t_j, r_{j,1} , \cdots , r_{j,n_j} \}$ for all $1 \le j \le N$.
Moreover, this can be achieved using only three of the primitives: initialisation, \textup{CNOT} and termination.
\end{corollary}
\begin{proof}
This corollary simply selects the QLNC solutions which have no classical control (apart from in the terminations).
\end{proof}

\section{Generalisation to qudits of prime dimension}
\label{qubitsection}

Classical network coding is not restricted to information streams consisting of individual bits.
Indeed, it is common in the literature to consider signals consisting of elements from finite fields in general, including the fields $\mathbb Z_d$ for $d$ a prime~\cite{LNCnew1,LNCnew2}.
As most proposed quantum hardware platforms involve operations on qubits, our work has focused mainly on linear network codes over $\mathbb Z_2$.
However, our techniques work equally well over qudits of any prime dimension $d$, using generalisations of Clifford group operations.

\subsection{Generalising the QLNC formalism}

On a Hilbert space of dimension $d$, label the standard basis states by $\ket{0}$, $\ket{1}$, \ldots, $\ket{d{-}1}$. 
Let $X_d$ and $Z_d$ be unitary operators satisfying $X \ket{a} = \ket{a{+}1 ~(\mathrm{mod}~d)}$ and $Z \ket{a} = \omega^a \ket{a}$ for $a \in \{0,1,\ldots,d{-}1\}$, where $\omega = \exp(2\pi i/d)$.
These operators are the basis of a definition of the generalised Pauli group on qudits of dimension $d$~\cite{GottesmanQudits,Appleby2005}.
The set of unitaries which preserves this extended Pauli group under conjugation corresponds to the Clifford group, and the effects of those operators on eigenstates of the $X$ and $Z$ operators can be simulated using an extension of the stabiliser formalism~\cite{GottesmanQudits,deBeaudrap2013}.

For the special case of $d \ge 2$ a prime, we may define an extension of the QLNC formalism to qudits of dimension $d$ by identifying the operations of the generalised Clifford group which correspond to the operations of the qubit QLNC formalism:
\begin{itemize}
\item
    Preparation of the states $\ket{0}$ and $\ket{+_d} = \tfrac{1}{\sqrt d}\bigl(\ket{0} + \ket{1} + \cdots + \ket{d{-}1} )$;
\item
    Performing (possibly classically controlled) $X_d$ and $Z_d$ operations on any qudit;
\item
    Measuring qudits in the eigenbasis of the $Z_d$ operator or the $X_d$ operator;
\item
    Addition operations $\mathrm{Add}_d$ on pairs of qudits which are connected in the network $G$, whose effect on standard basis states are $\mathrm{Add}_d \ket{x} \ket{y} = \ket{x}\ket{y + x ~(\mathrm{mod}~ d)}$.
\end{itemize}
We call circuits composed of these operations for a fixed $d \ge 2$, ``qudit QLNC circuits''.

These operations allow one to prepare states which are the analogue of parity function states, which one might call ``linear function states'', which have the form
\begin{equation}
  \label{eqn:parityFormulaExpana}
    \frac{1}{\sqrt{d^N}} \!
      \sum_{x \in \{0,1\}^N} \!\!
        \omega^{\phi(x)}\,
      \ket{f_1(x)} \otimes \ket{f_2(x)} \otimes \cdots \otimes \ket{f_n(x)}
\end{equation}~\\[-2ex]
for linear functions $f_k(a)$ and $\phi(a)$, and where $0 \le N \le n$.
We may represent $n$-qubit linear function states states by $(N{+}1)\times(n{+}2)$ ``linear function tableaus'' $T = {\bigl[\:\! C \:\!\big\vert\!\: \mathbf p \:\!\bigr]}$, which represent the linear function state by specifying the coefficients of the functions $f_k$ and $\phi$ in the same way as in the qubit case.
It is easy to show that preparation of the states $\ket{0}$ and $\ket{+_d}$ may be represented by the same column/row insertion steps, and the effects of the unitaries $X_d$, $Z_d$, and $\mathrm{Add}_d$ may be simulated in the same way through elementary column operations. (Indeed, one may use the same $\{0,1\}$-matrices in each case, albeit with coefficients modulo $d$.)
The procedures to simulate measurements are similar to the case $d = 2$, but must be described without relying (for instance) on $1$ being the only non-zero element.
It also becomes more helpful to describe the measurement outcomes as some element $s \in \mathbb Z_d$, representing a measurement of the $\omega^s$ eigenstate either of $X$ or of $Z$.
As before, we put the tableau into reduced row echelon form, making the $k\textsuperscript{th}$ column (counting from 0) a pivot column if possible, where $k$ is the qubit to be measured.
\begin{itemize}
\item
  For an $X_d$-eigenbasis measurement:
  \begin{enumerate}[label=({\alph*})]
  \item
    \label{item:XmeasProda}
    If $f_k(a) = a_g$ for an indeterminate which does not occur in any other qubit formula $f_j(a)$ --- \emph{i.e.},~if there is a single $1$ in the $g\textsuperscript{th}$ row of $C$ --- 
    then the state is unchanged by the measurement, and the measurement outcome is $s = p_g$.
  \item
    \label{item:XmeasEnta}
    Otherwise, let $z_{N\!\!\:{+}\!\!\:1}$ be a new indeterminate (represented in $C$ by adding a new row at the bottom), and choose a measurement outcome $s \in \mathbb Z_d$ uniformly at random.
    If $f_k(a)$ is constant prior to measurement, let $\Delta$ be the $(N{+}2)$-dimensional column vector with $1$ in the final row, and zero elsewhere; otherwise, if $f_k(a) = a_g$, let $\Delta$ be the $(N{+}2)$-dimensional column-vector with $-1$ in row $g$ and $1$ in row $N{+}1$ (counting from~$0$), and zero elsewhere.
    We then add $\Delta$ to the $k\textsuperscript{th}$ column of $C$, and subtract $s\Delta$ from $\mathbf p$.
  \end{enumerate}
\item
  For a $Z_d$-eigenbasis measurement:
  \begin{enumerate}[label=({\alph*})]
  \item
    \label{item:ZmeasProda}
    If $f_k(a) = c$ is a constant function, then the measurement leaves the state unchanged, and the measurement outcome is $c$.
  \item
    \label{item:ZmeasEnta}
    Otherwise, we select a measurement outcome $s \in \mathbb Z_d$ uniformly at random.
    Let $\Delta = s \mathbf e_0 - \mathbf c_k$.
    For any column $j$ of $T = {\bigl[\:\!C\:\!\big\vert\:\!\mathbf p\:\!\bigr]}$ which contains a non-zero coefficient $r_j \ne 0$ in the $g\textsuperscript{th}$ row (including the $k\textsuperscript{th}$ column itself), add $r_j \Delta$ to column $j$; then remove the row $g$ entirely from the tableau.
  \end{enumerate}
\end{itemize}
The analysis for these operations is similar to that of the case $d = 2$.
This allows us to simulate qudit QLNC circuits.

Finally, note that the property of parity function tableaus, that their ``parity function matrix'' $C$ has full rank, also holds for linear function tableaus for any $d$ prime, as these properties only depend on the fact that these matrices are defined over a field (which is a property on which we have also relied to consider reduced row echelons when simulating measurements).
As a result, those results (such as ``termination'' of qubits being well-defined) which rely on such rank properties are also true in the QLNC formalism on qudits of prime dimension.

It seems to us likely that, with some elaboration, the QLNC formalism may be extended to arbitrary dimensions $d \ge 2$.
However, such an elaboration must carefully accommodate the fact that not all non-zero coefficients are invertible modulo $d$.

\subsection{Entanglement distribution using the qudit QLNC formalism}

For qudits of prime dimension $d$, the natural analogues of Bell states and GHZ states are the states
\begin{equation}
\begin{aligned}
    \ket{\Phi_d^+} \,&=\, \frac{1}{\sqrt d} \,\sum_{x = 0}^{d-1} \,\ket{x}\ket{x};
  \qquad\qquad
    \ket{\mathrm{GHZ}_{d,n}} \,=\, \frac{1}{\sqrt d}\, \sum_{x=0}^{d-1}\, \underbrace{\ket{x}\ket{x}\cdots\ket{x}}_{\substack{\text{$n$ tensor} \\[0.25ex] \text{factors}}}\;.
\end{aligned}  
\end{equation}
These are evidently linear function states on qudits of dimension $d$.
As all of the QLNC formalism for qubits (including the notion of qubit termination) generalises in an appropriate way to qudits --- albeit possibly with a constant factor $d{-}1$ overhead, to realise some power of the $\mathrm{Add}_d$, $Z_d$, or $X_d$ operations --- we obtain the following results:
\begin{corollary}[to Theorem~\ref{thm:constdepth}]
Let $d$ be prime.
If a multiple multicast (including multiple unicast as a special case) classical $\mathbb Z_d$ linear network code exists over a network $G$, from a set of transmitter nodes $T= \{t_1 , \cdots , t_N\}$ with in-degree $0$ to a corresponding set of receiver nodes $R_j = \{r_{j,1} , \cdots , r_{j,n_j} \}$ with out-degree $0$, then there is a QLNC circuit whose \textup{CNOT} operations are located along the edges of $G$ and distributes $\ket{\Phi^+_d}$ and $\ket{\mathrm{GHZ}_{d,n_j}}$ states between corresponding transmitter/receiver node sets.
Moreover, this circuit has depth at most  ${2(d{-}1)\bigl[(\delta{+}2) (\chi{-}1) + 1\bigr]}$ time-steps, where $\delta$ is the largest in/out degree of any vertex in $G$, and $\chi$ is the chromatic number of $G$.
\end{corollary}

\begin{corollary}[to Theorem~\ref{mainthm1}]
  Let $d$ be prime.
Consider a multiple multicast (including multiple unicast as a special case) classical $\mathbb Z_d$ linear network code exists over any subgraph of the graph $G'= G \cup K_n$ sending a unit rate stream, where each edge of the graph is a unit rate bi-directional edge (but not allowing fractional routing).
Suppose that this code has a set of transmitting source vertices $T' = \{t_1, \ldots, t_{N'}\}$ for some $N' > 0$, where the first $N < N'$ of these have corresponding receiver sets $R_j = \{r_{j,1} , \cdots , r_{j,n_j} \}$ for $1 \le j \le N$ (with the remaining transmitters $t_{N{+}1}, \ldots, t_{N'}$ having signals which are not necessarily actually received by any ``receiver'' nodes).
Suppose further that \textup{(a)}~the information transmitted on the edges of $K_n$ from any single node, and \textup{(b)}~the information transmitted by the nodes $t_1$ through $t_N$, are linearly independent of each other.
Then by simulating this linear network code by a QLNC circuit, with \textup{CNOT} operations restricted to the same graph $G$ and classically-controlled Pauli operations oriented along the other edges, the resulting protocol generates a product of $\ket{\Phi^+_d}$ and $\ket{\mathrm{GHZ}_{d,n_j}}$ states, where each $\ket{\Phi^+_d}$ or $\ket{\mathrm{GHZ}_{d,n_j}}$ is over each of the sets $\{ t_j, r_{j,1} , \cdots , r_{j,n_j} \}$ for all $1 \le j \le N$.
\end{corollary}
The proofs of these statements are identical to the case $d= 2$, applying the extension of the QLNC formalism to $d > 2$ dimensional qudits.

\section{Remarks on computational complexity}
\label{comp}

We now consider the computational complexity of the QLNC formalism, and also remark on the complexity of finding linear network codes.

\subsection{Comparison of the QLNC formalism to the stabiliser formalism}

Recall that a parity function tableau  on $n$ qubits is a matrix of size $(N{+}1)\times(n{+}2)$, where $0 \le N \le n$ is some number of indeterminates involved in the expression of the state.
As every parity function tableau has the same first column, the amount of information can be bounded above by $(N{+}1) \times (n{+}1)$ bits.
By allocating enough space for an $(n{+}1)\times(n{+}1)$ matrix, and by maintaining lists to record which rows and columns in this space are actually occupied, we suppose that the data structure used for the matrix allows for $O(1)$ time  row and column insertion and removal, apart from the time required to actually initialise the entries of new rows or columns.

Several of the QLNC circuit operations may be represented by very simple operations or transformations on the tableau:
\begin{itemize}
\item
    Preparation of a fresh qubit involves introducing a new row and a new column, which involves $O(n + N)$ time to initialise.
\item
    Performing a single CNOT, $X$, or $Z$ gate involves an elementary row operation, which requires $O(N) \subseteq O(n)$ time to perform.
\end{itemize}
Others of the operations are somewhat more involved:
\begin{itemize}
\item
    Performing measurements --- destructive or otherwise --- involves first putting the parity function tableau into a reduced row echelon form, which requires $O(N^2 n)$ time.
    This dominates the run-time required for the remaining operations:
    \begin{itemize}
    \item 
        For an $X$ measurement, the subsequent operations may involve adding a new row, which takes $O(n + N)$ time; and adding a vector of size $O(N)$ to two columns, which takes $O(N)$ time.
    \item
        For a $Z$ measurement, the subsequent operations may involve adding a column vector of size $O(N)$ to $O(n)$ columns, and removing a row and a column, which all together takes $O(Nn)$ time.
    \end{itemize}
\item
    Terminating a qubit requires a measurement, and also  an appropriate set of qubits on which to perform $Z$ operations.
    Finding the latter also involves putting the tableau in reduced row echelon form, and $O(N)$ further work to determine an appropriate correction set; thus this also takes time $O(N^2 n)$.
\end{itemize}

A natural comparison to make is with the stabiliser formalism~\cite{Aaronson2004}.
This also requires $O(n^2)$ space, with little room for improvement beyond techniques to represent sparse matrices.
Preparation of a fresh qubit in the stabiliser formalism similarly involves extending the matrix, and takes $O(n)$ time; realising a CNOT, $X$, or $Z$ operation takes time $O(n)$.
Using na\"{\i}ve techniques, simulating a measurement in the Stabiliser formalism may take time $O(n^3)$, involving Gaussian elimination; this may be avoided using ``destabiliser'' methods~\cite{Aaronson2004} at the cost of doubling the size of the tableau to take $O(n^2)$ time.
In the worst case where $N \in \Theta(n)$, the run-time bounds for the QLNC formalism are then worse than those involving ``destabiliser'' methods; but for any circuit in which there is a bound $N \in o(n^{1/2})$, we obtain a better than constant factor improvement in the complexity of measurement, and a better than quadratic improvement in the complexity of simulating CNOT, $X$, and $Z$ gates.

In summary, the computational advantage of the QLNC formalism, when simulating QLNC circuits, is the ability to take advantage of the potential for the parity function tableau to occupy space ${}\! \ll n^2$, in which case the operations required to transform it are less computationally costly.
Even in the worst case, the fact that parity function tableaus  have size $n^2 + O(n)$, rather than size $2n^2 + O(n)$, will also yield a mild improvement in performance for unitary operations (while incurring a performance hit for measurements).

\subsection{On the complexity of finding QLNC circuits for entanglement distribution problems}

Here, we consider the complexity of \emph{finding} a QLNC to perform a particular entanglement distribution task in a given network $G$.

It is clear that when a linear network code for the classical $k$-pairs (or multiple multicast) problem exists in a particular network $G$, we may easily convert this to a constant depth QLNC circuit to solve the corresponding entanglement distribution problem on a quantum platform with the same interaction topology $G$ (with the mild restriction that nodes are either transmitters or receivers or relays, as previously discussed).
However, it is not always easy to find such a linear network code.
Lehman and Lehman \cite{Lehman2004} show that deciding whether a network code exists is NP-hard in general.
As Kobayashi \textit{et al} \cite{Kobayashi2009} note, the $k$-pair problem is thus itself NP-hard, as all network coding can be reduced to an instance of the $k$-pair problem~\cite{Dougherty2006}.

Given that finding network codes is a hard problem in general, it is reasonable to ask whether reducing the problem of entanglement distribution to the problem of finding linear network codes is of any practical advantage.
One answer to this is that the problem of classical network coding has already received significant attention (\emph{e.g.}, \cite{Dougherty2006, LNCnov1, LNCnov2, LNCnov3, LNCnov4}), and thus such a reduction enables existing results and understanding to be transferred to the problem of entanglement distribution.
Furthermore, the existing proof of NP-hardness appears to require a somewhat specialised network architecture in principle.
(To us, this seems to mirror the situation with the bounds on the depth of the ``constant-depth'' QLNC circuits described in Theorem~\ref{thm:constdepth}: while the bound depends on parameters such as vertex-chromatic number which are NP-hard to compute in general, in many practical examples they may be computed very easily.)
Finally, as we allow unconstrained classical control in QLNCs (\emph{i.e.},~the classical control could be thought of as being transmitted through a completely connected graph, as in Section~\ref{eswap}), we should expect it to be easier to find a QLNC for a network $G$, and perhaps to sometimes find a QLNC for entanglement distribution where there is no solution to the corresponding classical linear network coding problem.

In any case, as our results more generally allow an edge to be used more than once, it is not clear whether we should expect the problem of finding QLNC solutions to entanglement distribution to be similar to that of solving the $k$ pairs problem.
The complexity of this problem is open; though from our results, it is clear that it cannot be worse than NP-hard.
We conjecture that it should be possible to do so in polynomial time.

\section{Proof of Theorem~\ref{mainthm1}}
\label{app1}

Finally, we present here a more thorough presentation of the proof of Theorem~\ref{mainthm1}. In particular, we adopt a more concrete presentation in the hopes of describing in some detail what transformations of the states involved.

Let there be $n$ qubits, ordered such that the first $n_1$ are prepared in the $\ket{+}$ state and the remaining $n_2 = n - n_1$ are prepared in the $\ket{0}$ state. The QLNC circuits described consist of four primitives: initialisation (i.e., preparation of qubits in the $\ket{+}$ or $\ket{0}$ state, as stated directly above); CNOT gates; measurements which can classically control Pauli-$X$ gates on other qubits; and termination operations. Firstly, we note that the principle of deferred measurement can be used to express an equivalent circuit with the classically controlled $X$ gates replaced by CNOT gates and deferred measurement on the control qubit, as shown in Fig.~\ref{f004}(a) and (b), and secondly, we address a generalised version of the circuit in question, as shown in Fig.~\ref{f004}(c). The remainder of the proof proceeds as follows: firstly, we relate the state of this generalised circuit after the network of CNOT gates to the actual circuit we want to express; secondly, we prove by induction that the state is correctly simulated by the QLNC formalism as the individual CNOT gates are executed; thirdly we show that the termination process does indeed remove qubits as required, without disturbing the rest of the state; and finally we show that the desired product of $\ket{\Phi^+}$ and $\ket{\mathrm{GHZ}}$ states is only realised if and only if the measurements do not reveal information thereabout, and that the Gaussian elimination procedure described is necessary and sufficient to verify this.\\
\begin{figure}[!t]
	\centering
	\includegraphics[width=0.58\textwidth]{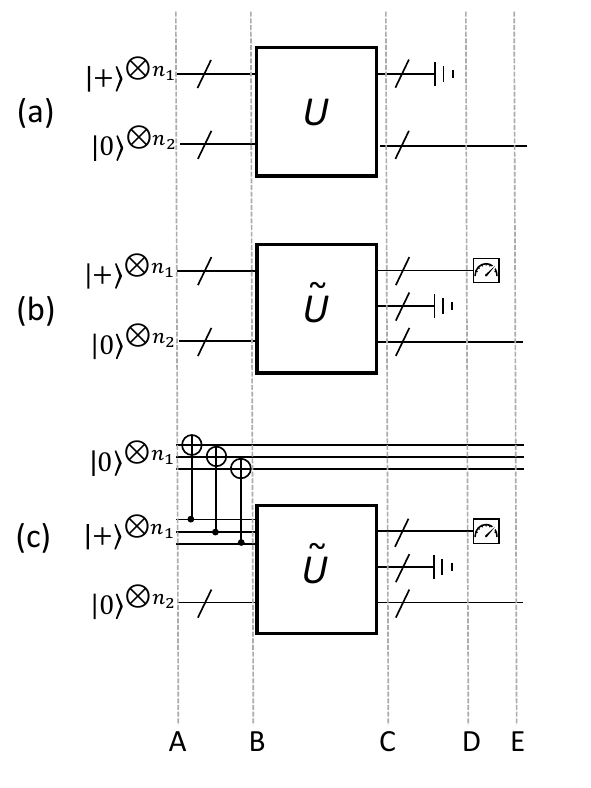}
	\captionsetup{width=0.95\linewidth}
	\caption{\small{Illustration of equivalent circuits used in the proof of Theorem~\ref{mainthm1}, the three parallel vertical lines of descending size (i.e., a rotated `earth' symbol, as used in electrical engineering) denotes termination: (a) shows the actual quantum circuit, consisting of qubits initialised in the $\ket{+}$ and $\ket{0}$ states, a network of $CNOT$ gates, and measurements classically controlling Pauli-$X$ gates, and  some terminations; (b) shows the same circuit, but now with deferred measurement (such that the classically controlled Pauli-$X$ gates can be represented as $CNOT$ gates; and (c) shows the circuit with additional ancilla qubits entangled with the qubits initialised in the $\ket{+}$ state, as is required in the proof.}}
	\label{f004}
\end{figure}
\indent Fig.~\ref{f004} illustrates a general instance of the circuit, in which $U$, in Fig.~\ref{f004}(a), is a block consisting of CNOT gates and classically controlled Pauli-$X$ gates, in Fig.~\ref{f004}(b) the principle of deferred measurement is used to draw an equivalent circuit, with CNOT gates replacing Pauli-$X$ in a block now labelled $\tilde{U}$, with measurements deferred until the end of the circuit. This allows us to write down the state directly after $\tilde{U}$:
\begin{align}
    \ket{\psi_C} & = \tilde{U}(\ket{+}^{\otimes n_1}  \ket{0}^{\otimes n_2} ) \nonumber \\
    \label{app1eq10}
    & =  \frac{1}{\sqrt{2^{n_1}}}\sum_{i=0}^{2^{n_1}-1} \tilde{U}(\ket{i}  \ket{0}^{\otimes n_2} ),
\end{align}
where $i$ is binary number, later in the analysis we use $\mathbf{i}$ as the binary vector corresponding to the binary number $i$ (i.e., the $j^{th}$ element of $\mathbf{i}$ is the $j^{th}$ digit of $i$) and we use each of $\ket{i}$ and $\ket{\mathbf{i}}$ to denote the corresponding $n_q$-qubit quantum state, where $n_q$ is the number of digits in $i$ (and therefore the number of elements in $\mathbf{i}$). In the analysis, it is helpful to consider the circuit in Fig.~\ref{f004}(c), in which $n_1$ ancilla qubits are prepended to the state. Each ancilla is initialised in the $\ket{0}$ state, and then is the target of a CNOT by one of the qubits initialised in the $\ket{+}$ states (that is, a different one of these qubits controls the CNOT for each of the ancillas). This allows the state before $\tilde{U}$ to be expressed:
\begin{equation}
    \label{app1eq15}
    \ket{\tilde{\psi}_B} = \frac{1}{\sqrt{2^{n_1}}}\sum_{i=0}^{2^{n_1}-1} \ket{i}\ket{i}  \ket{0}^{\otimes n_2} .
\end{equation}
which in turn allows us to express the state after $\tilde{U}$:
\begin{equation}
\label{app1eq20}
    \ket{\tilde{\psi}_C} = \frac{1}{\sqrt{2^{n_1}}}\sum_{i=0}^{2^{n_1}-1} \ket{i} \tilde{U}(\ket{i}  \ket{0}^{\otimes n_2} ).
\end{equation}
These extra ancillas have been introduced to make it easier to keep track of the state, and we later rely on the correspondence between (\ref{app1eq10}) and (\ref{app1eq20}) to show that these additional ancillas are indeed just an analytical device and do not affect the validity of the simulation of the actual circuit in the formalism.\\
%
%for later use we define $\mathbf{i}$ as the corresponding vector (i.e., the $j^{th}$ element of $\mathbf{i}$ is the $j^{th}$ digit of $i$).\\
%
\indent We can now introduce the QLNC formalism in vectorised form. We order the vector such that the first $n_1$ elements correspond to the qubits initialised in the $\ket{+}$ state, and are therefore labelled with a unique symbol in the initialisation process. Furthermore, the effect of each of these performing a CNOT on (a distinct) one of the ancillas is to copy the label to the ancilla, therefore it is convenient to think of the ancillas as having been labelled. Let these labels be  $a_1 \cdots a_{n_1}$, and the actual qubits be labelled $q_1 ... q_n$ (which in general will be sums over the terms  $a_1  \cdots  a_{n_1}$). Stacking these up into vectors, we have that $\mathbf{a} = [a_1 , \cdots , a_{n_1}]^T$ and $\mathbf{q} = [q_1 , \cdots , q_{n}]^T$, such that:
\begin{equation}
\label{app1eq30}
    \mathbf{q} = \mathrm{L} \mathbf{a},
\end{equation}
where $\mathrm{L}$ is a $n \times n_1$ binary matrix which keeps track of how the labels of the various qubits are related to the ancilla labels, i.e., initially $\mathrm{L} = [\mathbbm{1} | \mathbf{0}]^T$. In the QLNC formalism, the operation of a CNOT with the $j^{th}$ qubit controlling the $k^{th}$ qubit, is that the $j^{th}$ row of $\mathrm{L}$ is added to the $k^{th}$ row (modulo-2), that is $\mathrm{L}_{k,*} \leftarrow \mathrm{L}_{k,*} + \mathrm{L}_{j,*}$ (here `$*$' means the entirety of that row). Moving on to the network of CNOT gates (including those which have been included by the deferred measurement equivalence), we prove by induction that the quantum state is in the form:
\begin{equation}
\label{app1eq40}
    \ket{\tilde{\psi}_{BC}} = \frac{1}{\sqrt{2^{n_1}}}\sum_{i=0}^{2^{n_1}-1} \ket{i}   \ket{\mathrm{L}\mathbf{i}} ,
\end{equation}
where $\ket{\tilde{\psi}_{BC}}$ is the quantum state at an arbitrary point between $\ket{\tilde{\psi}_{B}}$ and $\ket{\tilde{\psi}_{C}}$ (i.e., within the block of CNOT gates, $\tilde{U}$).\\
%
%$\ket{\phi_{\ket{i}}}$ is the state of the $n$ qubits that is entangled with the specific value of $\ket{i}$, and is such that the $j^{th}$ digit of $\ket{\phi_{\ket{i}}}$ is equal to $\mathrm{L}_{j,*} \mathbf{i}$ (where $\mathbf{i}$ is the vector that corresponds to $\ket{i}$, as previously defined).\\
%
\indent For the inductive proof, we observe that the initial definition of $\mathrm{L}$ (i.e., in the text below (\ref{app1eq30}) is of a format that corresponds to this definition, i.e., for the initial state in (\ref{app1eq20}). Turning to how the quantum state is changed by a CNOT gate, to simplify the notation (and without loss of generality) we re-order the qubits (and therefore the rows of $\mathrm{L}$ and $\mathbf{q}$) such that the first qubit is the control, and the second the target, before the CNOT we have:
\begin{align}
    \ket{\tilde{\psi}_{BC}} = & \frac{1}{\sqrt{2^{n_1}}} \left( \sum_{i\, s.t. \mathrm{L}_{1,*}\mathbf{i} = 0 , \, \mathrm{L}_{2,*}\mathbf{i} = 0 } \ket{i} \ket{00}  \ket{\psi_i'} + \sum_{i\, s.t. \mathrm{L}_{1,*}\mathbf{i} = 0 , \, \mathrm{L}_{2,*}\mathbf{i} = 1 } \ket{i} \ket{01}  \ket{\psi_i''} \right. \nonumber \\
    \label{app1eq50}
    & \left. \, \, \, \, \, \, \, \, \, \, \, \, \, \, \, \, \, \, \, \, \, \, \, \, \,  + \sum_{i\, s.t. \mathrm{L}_{1,*}\mathbf{i} = 1 , \, \mathrm{L}_{2,*}\mathbf{i} = 0 } \ket{i} \ket{10}  \ket{\psi_i'''} + \sum_{i\, s.t. \mathrm{L}_{1,*}\mathbf{i} = 1 , \, \mathrm{L}_{2,*}\mathbf{i} = 1 } \ket{i} \ket{11}  \ket{\psi_i''''} \right),
\end{align}
where $s.t.$ means `such that', and $\ket{\psi_i'}$, $\ket{\psi_i''}$, $\ket{\psi_i'''}$ and $\ket{\psi_i''''}$ represent the remainder of the quantum state in each term, which is not required for this analysis. After the performing a CNOT gate on the first two qubits we have:
\begin{align}
    & \left(\textnormal{CNOT} \otimes \id_{n-2}\right)\ket{\tilde{\psi}_{BC}}  \nonumber \\
    & \,\,\,\,\, =  \frac{1}{\sqrt{2^{n_1}}} \left( \sum_{i\, s.t. \mathrm{L}_{1,*}\mathbf{i} = 0 , \, \mathrm{L}_{2,*}\mathbf{i} = 0 } \ket{i} \ket{00}  \ket{\psi_i'} + \sum_{i\, s.t. \mathrm{L}_{1,*}\mathbf{i} = 0 , \, \mathrm{L}_{2,*}\mathbf{i} = 1 } \ket{i} \ket{01}  \ket{\psi_i''} \right. \nonumber \\
     & \,\,\,\,\, \left. \, \, \, \, \, \, \, \, \, \, \, \, \, \, \, \, \,   + \sum_{i\, s.t. \mathrm{L}_{1,*}\mathbf{i} = 1 , \, \mathrm{L}_{2,*}\mathbf{i} = 0 } \ket{i} \ket{11}  \ket{\psi_i'''} + \sum_{i\, s.t. \mathrm{L}_{1,*}\mathbf{i} = 1 , \, \mathrm{L}_{2,*}\mathbf{i} = 1 } \ket{i} \ket{10}  \ket{\psi_i''''} \right) \nonumber \\
    & \,\,\,\,\, = \frac{1}{\sqrt{2^{n_1}}} \left( \sum_{i\, s.t. \mathrm{L}_{1,*}\mathbf{i} = 0 , \, (\mathrm{L}_{1,*} + \mathrm{L}_{2,*})\mathbf{i} = 0 } \ket{i} \ket{00}  \ket{\psi_i'} + \sum_{i\, s.t. \mathrm{L}_{1,*}\mathbf{i} = 0 , \, (\mathrm{L}_{1,*} + \mathrm{L}_{2,*})\mathbf{i} = 1 } \ket{i} \ket{01}  \ket{\psi_i''} \right. \nonumber \\
    \label{app1eq60}
    & \,\,\,\,\, \left. \, \, \, \, \, \, \, \, \, \, \, \, \, \, \, \, \,  + \sum_{i\, s.t. \mathrm{L}_{1,*}\mathbf{i} = 1 , \, (\mathrm{L}_{1,*} + \mathrm{L}_{2,*})\mathbf{i} = 1 } \ket{i} \ket{11}  \ket{\psi_i'''} + \sum_{i\, s.t. \mathrm{L}_{1,*}\mathbf{i} = 1 , \, (\mathrm{L}_{1,*} + \mathrm{L}_{2,*})\mathbf{i} = 0 } \ket{i} \ket{10}  \ket{\psi_i''''} \right), 
\end{align}
which we can see is consistent with the operation of a CNOT where the first qubit controls the second in the QLNC formalism, i.e., the assignment $\mathrm{L}_{2,*} \leftarrow \mathrm{L}_{2,*} + \mathrm{L}_{1,*}$, thereby completing the inductive proof. It is worth observing that, while our proposed formalism was conceptualised from the starting point of classical network codes (as emphasised in Corollary~\ref{corr}), the manner in which the state is tracked bears some resemblance to the quadratic representation of the Stabliser formalism as described by Dehaene and de Moor \cite{dehaene}.\\
\indent $\ket{\tilde{\psi}_C}$ is simply $\ket{\tilde{\psi}_{BC}}$ after all of the CNOT gates in $\tilde{U}$ have been executed, and using the correspondence between (\ref{app1eq10}) and (\ref{app1eq20}) allows us to express $\ket{\psi_C}$ from (\ref{app1eq40}):
\begin{equation}
\label{app1eq61}
    \ket{\psi_{C}} = \frac{1}{\sqrt{2^{n_1}}}\sum_{i=0}^{2^{n_1}-1}   \ket{\mathrm{L}\mathbf{i}} ,
\end{equation}
\indent The next step in the circuit is the termination of any qubit which is left such that its label is the sum of two or more symbols, and indeed any other qubits with a single symbol label if desired. In the termination process the goal is, for any given post measurement state, that the corresponding qubits should be measured out in such a way that the superposition of quantum states is the same for the rest of the qubits, whichever state was measured. That is, if the state is $\ket{0}\ket{\phi} + \ket{1}\ket{\tilde{\phi}}$, termination and removal of the first qubit should leave the state $\ket{\phi} + \ket{\tilde{\phi}}$. This can be achieved in one of three ways: firstly, if the qubit to be terminated has a label which can be expressed exactly as a sum of qubits that are measured, then it can be measured out directly, as no additional information will be learned by doing so (in reality this measurement will have already taken place, although for the analysis this is treated as a deferred measurement, but this does not affect the validity of measuring it out directly). Conversely, in the case where the label of the qubit to be terminated is linearly independent of all of the other qubit labels, then it can also be measured out, as this will not reveal any information about the entangled superposition of interest. To see this, WLoG we re-order the qubits (including the ancilla qubits) such that the first qubit is to be terminated, from which we can express the state:
\begin{equation}
\label{app1eq75}
    \ket{\psi_{CD}}= \frac{1}{\sqrt{2^{n_1}}}\sum_{i=0}^{2^{n_1}-1}   \ket{\mathrm{L}_{1,*}\mathbf{i}}\ket{ \mathrm{L}_{2:n,*}\mathbf{i}},
\end{equation}
and thus we can see that because $\mathrm{L}_{1,*}$ is linearly independent of all of the other rows of $\mathrm{L}$, measuring it out will not collapse any other terms in the superposition.\\
%
%the presence of the term $L_{1,1}\mathbf{i}_1$ in the first ket means that measuring the first qubit will impart no information about the superposition of the remaining $n-1$ qubits.\\
%
%and it has the unique label corresponding to the first ancilla label, $a_1$ (and that it is the first qubit to be terminated -- this is simply for convenience in expressing the normalisation of the superposition). In this case, the matrix $L$ is of the form:
%\begin{equation}
%    L = \left[\begin{array}{c|c}
%1 & L_{1,2:n_1} \\
%\hline
%\vdots & \\
%0 & L_{2:n,2:n_1} \\
%\vdots & \\
%\end{array}\right],
    %\left[ \begin{matrix} 1 & L_{1,2:n_1} \\ 0 \right]
%\end{equation}
%where $i=0 | \text{meas}(m)$ denotes that only values of $i$ that are consistent with the $m^{th}$ measurement are included. By definition, in (\ref{app1eq75}) for all $i$, $\ket{\tilde{\phi}_{\ket{i}}^{(m)}}$ and $\ket{\phi_{\ket{i}}^{(m)}}$ only differ in the place of the qubit that is being measured, therefore the measurement does not reveal any useful information about the superposition (except insofar as that they are entangled with different values of the ancillas $\ket{i}$, but this is remedied in due course.\\
%
\indent So we move onto the third option for termination, where the qubit to be terminated can be expressed as a sum of qubit labels, of which at least some haven't been measured. Once again, for simplicity of exposition and without loss of generality, we consider that it is the first qubit, labelled $q_1$, that is to be terminated. To see how the termination process works, first let us write the linear expression of $q_1$ in terms of the other qubit labels: $q_1 = \mathbf{r}^T\mathbf{q}_{2:n}$, where $\mathbf{r}$ is a binary vector that selects the other qubits whose labels sum to $q_1$. We now express $\mathbf{r} = \mathbf{r}_a + \mathbf{r}_b$, such that $\mathbf{r}_a$ corresponds to qubits that are measured, and $\mathbf{r}_b$ corresponds to qubits that are not measured. Thus we can re-express (\ref{app1eq61}), noting that $\mathbf{q}_{2:n} = \mathrm{L}_{2:n_1,*}\mathbf{a}$, from (\ref{app1eq30}):
\begin{equation}
\label{app1eq90}
    \ket{\psi_{CD}} = \frac{1}{\sqrt{2^{n_1}}} \left( \sum_{i  \, s.t. \, (\mathbf{r}_a + \mathbf{r}_b )^T \mathrm{L}_{2:n_1,*} \mathbf{i} = 0}  \ket{0} \ket{\mathrm{L}_{2:n_1,*} \mathbf{i}} + \sum_{i  \, s.t. \, (\mathbf{r}_a + \mathbf{r}_b)^T \mathrm{L}_{2:n_1,*} \mathbf{i} = 1}  \ket{1} \ket{\mathrm{L}_{2:n_1,*} \mathbf{i}} \right),
\end{equation}
taking first the case where the existing measurements are such that $\mathbf{r}_a^T \mathrm{L}_{2:n_1,*}=0$,  (\ref{app1eq90}) becomes:
\begin{equation}
\label{app1eq95}
    \ket{\psi_{CD}} = \frac{1}{\sqrt{2^{n_1}}} \left( \sum_{i  \, s.t. \, \mathbf{r}_b T \mathrm{L}_{2:n_1,*} \mathbf{i} = 0}  \ket{0} \ket{\mathrm{L}_{2:n_1,*} \mathbf{i}} + \sum_{i  \, s.t. \, \mathbf{r}_b^T \mathrm{L}_{2:n_1,*} \mathbf{i} = 1}  \ket{1} \ket{\mathrm{L}_{2:n_1,*} \mathbf{i}} \right).
\end{equation}
Next, we treat the $X$ observable measurement in the equivalent form of a Hadamard gate, followed by a $Z$ observable (computational basis) measurement. Thus, the Hadamard gate transforms (\ref{app1eq95}) to:
\begin{equation}
\label{app1eq100}
    \ket{\psi_{CD}} = \frac{1}{\sqrt{2^{n_1+1}}} \left( \sum_{i  \, s.t. \, \mathbf{r}_b^T \mathrm{L}_{2:n_1,*} \mathbf{i} = 0}  (\ket{0}+\ket{1}) \ket{\mathrm{L}_{2:n_1,*} \mathbf{i}} + \sum_{i  \, s.t. \, \mathbf{r}^T_b \mathrm{L}_{2:n_1,*} \mathbf{i} = 1}  (\ket{0}-\ket{1}) \ket{\mathrm{L}_{2:n_1,*} \mathbf{i}} \right).
\end{equation}
After which the state is measured, and so we must address each of the cases where we measure each of $0$ or $1$. In the former we can see that the state collapses to:
\begin{equation}
\label{app1eq100a}
    \ket{\psi_{CD}} = \frac{1}{\sqrt{2^{n_1+1}}} \left( \sum_{i  \, s.t. \, \mathbf{r}_b^T \mathrm{L}_{2:n_1,*} \mathbf{i} = 0}  \ket{0} \ket{\mathrm{L}_{2:n_1,*} \mathbf{i}} + \sum_{i  \, s.t. \, \mathbf{r}^T_b \mathrm{L}_{2:n_1,*} \mathbf{i} = 1}  \ket{0} \ket{\mathrm{L}_{2:n_1,*} \mathbf{i}} \right),
\end{equation}
with the terminated qubit still included. Whereas if we measure $1$, we get:
\begin{equation}
\label{app1eq100b}
    \ket{\psi_{CD}} = \frac{1}{\sqrt{2^{n_1+1}}} \left( \sum_{i  \, s.t. \, \mathbf{r}_b^T \mathrm{L}_{2:n_1,*} \mathbf{i} = 0}  \ket{1} \ket{\mathrm{L}_{2:n_1,*} \mathbf{i}} - \sum_{i  \, s.t. \, \mathbf{r}^T_b \mathrm{L}_{2:n_1,*} \mathbf{i} = 1}  \ket{1} \ket{\mathrm{L}_{2:n_1,*} \mathbf{i}} \right),
\end{equation}
However, by definition, zero is measured when there are an even number of ones in $\mathbf{r}_b \odot (\mathrm{L}_{2:n_1,*} \mathbf{i})$ and one is measured when there are an odd number of ones therein (where $\odot$ means element-wise multiplication). Therefore, applying a Pauli-$Z$ (phase) gate to each qubit which corresponds to a one in $\mathbf{r}$ guarantees the correct adjustment, and this is exactly what is prescribed in the termination process. Thus, after the correction (\ref{app1eq100b}) becomes:
\begin{equation}
\label{app1eq100bb}
    \ket{\psi_{CD}} = \frac{1}{\sqrt{2^{n_1+1}}} \left( \sum_{i  \, s.t. \, \mathbf{r}_b^T \mathrm{L}_{2:n_1,*} \mathbf{i} = 0}  \ket{1} \ket{\mathrm{L}_{2:n_1,*} \mathbf{i}} + \sum_{i  \, s.t. \, \mathbf{r}^T_b \mathrm{L}_{2:n_1,*} \mathbf{i} = 1}  \ket{1} \ket{\mathrm{L}_{2:n_1,*} \mathbf{i}} \right),
\end{equation}
\indent Turning now to the alternative situation, where the existing measurements are such that $\mathbf{r}_a^T \mathrm{L}_{2:n_1,*}=1$,  (\ref{app1eq90}) becomes:
\begin{equation}
\label{app1eq95a}
    \ket{\psi_{CD}} = \frac{1}{\sqrt{2^{n_1}}} \left( \sum_{i  \, s.t. \, \mathbf{r}_b T \mathrm{L}_{2:n_1,*} \mathbf{i} = 0}  \ket{1} \ket{\mathrm{L}_{2:n_1,*} \mathbf{i}} + \sum_{i  \, s.t. \, \mathbf{r}_b^T \mathrm{L}_{2:n_1,*} \mathbf{i} = 1}  \ket{0} \ket{\mathrm{L}_{2:n_1,*} \mathbf{i}} \right).
\end{equation}
which the Hadamard gate transforms to:
\begin{equation}
\label{app1eq100n}
    \ket{\psi_{CD}} = \frac{1}{\sqrt{2^{n_1+1}}} \left( \sum_{i  \, s.t. \, \mathbf{r}_b^T \mathrm{L}_{2:n_1,*} \mathbf{i} = 0}  (\ket{0}-\ket{1}) \ket{\mathrm{L}_{2:n_1,*} \mathbf{i}} + \sum_{i  \, s.t. \, \mathbf{r}^T_b \mathrm{L}_{2:n_1,*} \mathbf{i} = 1}  (\ket{0}+\ket{1}) \ket{\mathrm{L}_{2:n_1,*} \mathbf{i}} \right).
\end{equation}
After which the state is measured, and so we must address each of the cases where we measure each of $0$ or $1$. In the former we can see that the state again collapses to:
\begin{equation}
\label{app1eq100an}
    \ket{\psi_{CD}} = \frac{1}{\sqrt{2^{n_1+1}}} \left( \sum_{i  \, s.t. \, \mathbf{r}_b^T \mathrm{L}_{2:n_1,*} \mathbf{i} = 0}  \ket{0} \ket{\mathrm{L}_{2:n_1,*} \mathbf{i}} + \sum_{i  \, s.t. \, \mathbf{r}^T_b \mathrm{L}_{2:n_1,*} \mathbf{i} = 1}  \ket{0} \ket{\mathrm{L}_{2:n_1,*} \mathbf{i}} \right),
\end{equation}
with the terminated qubit still included. Whereas if we measure $1$, we get:
\begin{equation}
\label{app1eq100bn}
    \ket{\psi_{CD}} = \frac{1}{\sqrt{2^{n_1+1}}} \left( -\sum_{i  \, s.t. \, \mathbf{r}_b^T \mathrm{L}_{2:n_1,*} \mathbf{i} = 0}  \ket{1} \ket{\mathrm{L}_{2:n_1,*} \mathbf{i}} + \sum_{i  \, s.t. \, \mathbf{r}^T_b \mathrm{L}_{2:n_1,*} \mathbf{i} = 1}  \ket{1} \ket{\mathrm{L}_{2:n_1,*} \mathbf{i}} \right),
\end{equation}
However, applying the same correction as before, we get:
\begin{equation}
\label{app1eq100bbn}
    \ket{\psi_{CD}} = -\frac{1}{\sqrt{2^{n_1+1}}} \left( \sum_{i  \, s.t. \, \mathbf{r}_b^T \mathrm{L}_{2:n_1,*} \mathbf{i} = 0}  \ket{1} \ket{\mathrm{L}_{2:n_1,*} \mathbf{i}} + \sum_{i  \, s.t. \, \mathbf{r}^T_b \mathrm{L}_{2:n_1,*} \mathbf{i} = 1}  \ket{1} \ket{\mathrm{L}_{2:n_1,*} \mathbf{i}} \right),
\end{equation}
So we can see that, regardless of the previous measurement outcomes, and the outcome of the $X$-basis measurement of the qubit being terminated, we get the same quantum state up to an unobservable global phase.\\
%
%\indent In general it may be the case that after a given termination, the next termination can no longer be written as a linear combination of the remaining qubits, and thus can be measured out directly. After terminating all of the qubits which are symbolically labelled as sums (as well as any undesired single label qubits, which can be terminated in the same way), the state will consist solely of singly labelled qubits.\\ 
%
\indent Following the layer of terminations, we can express the state as:
\begin{equation}
    \label{app1eq110}
    \ket{\psi_{D}} = \frac{1}{\sqrt{2^{n'_1}}}\sum_{i=0}^{2^{n_1'}-1}   \ket{\mathrm{L}\mathbf{i}} ,
\end{equation}
where we have omitted any qubits that have been measured out and any labels that are no longer present in any qubit label. Thus rows in $\mathrm{L}$ will have either exactly one, or more than one element equal to one (and the rest equal to zero, as it is a binary matrix). We know that all rows of $\mathrm{L}$ with multiple elements equal to one are measured (i.e., otherwise they would have been terminated), and in general some rows with exactly one element equal to 1 may be measured too. To verify that none of these measurements imparts information that would collapse the superposition of interest, we follow the same rationale as that described around (\ref{app1eq75}). Specifically, we construct a $n \times n_M$ matrix (where $n_M$ is the number of measurements), $\mathrm{M}$ such that each row corresponds to one measurement. For example, if we have five symbols in total, $a_1 \cdots a_5$, and we measure a qubit labelled $a_1 \oplus a_4$, the corresponding row of $\mathrm{M}$ would be $[1,0,0,1,0]$. We now re-order the columns of $\mathrm{M}$ such that the first $n_r$ correspond to symbols that aren't present in the final entangled state, and perform Gaussian elimination such that the matrix is in upper-echelon form, let this transformed version of $\mathrm{M}$ be denoted $\mathrm{M}'$. A necessary and sufficient condition for the measurements not to have imparted any information that collapses the final entangled state is that each row which is not all zeros should have at least one element equal to one in the first $n_r$ columns. An example of $\mathrm{M}'$ is shown in (\ref{mateq}), and the necessary and sufficient condition essentially means that the label of each measured qubit includes at least one unique symbol, not present in any other label (either those of other measured qubits, or in the labels of the qubits that compose the final state), and thus, by the same reasoning given in and around (\ref{app1eq75}) means that the final state will not be collapsed by this measurement.
%
%It is for this reason that the Gaussian elimination process described is necessary: let $\mathrm{M}$ be the matrix of measurements after the Gaussian elimination process, with the first $r$ columns corresponding to the symbols that are not present in the final state -- the requirement that each row after the $r^{th}$ are all zeros guarantees that there is a maximum of $r$ independent measurements, and the condition that there is at least one non-zero element in the first $r$ columns for any rows which aren't all zero means that a unique symbol that \textit{is not} present in the final state \textit{is} present in any given measurement, when the measurement matrix is expressed in upper-echelon form. An example of the matrix $\mathrm{M}$ is shown in (\ref{mateq}).\\
%

%
\begin{equation}
\label{mateq}
  \mathrm{M}' =   \left. \left[
\begin{array}{c c c c c c}
\bovermat{$n_r$ cols}{1 & \cdots & & & &  \cdots} \\
0 & 1 & & & & \cdots\\
\vdots & \ddots & 1 & & & \cdots \\
& & & 1 & & \cdots \\
& & &  & \ddots & \cdots \\
\end{array}
\right] \right\} \text{\scriptsize{$n_M$ rows}} \normalsize .
\end{equation}
\indent Having performed these measurements, and verified the condition of not imparting information that collapses the superposition, we have the final state:
\begin{equation}
    \label{app1eq120}
    \ket{\psi_{E}} = \frac{1}{\sqrt{2^{n''_1}}}\sum_{i=0}^{2^{n_1''}-1}   \ket{\mathrm{L}\mathbf{i}} ,
\end{equation}
where each row of $\mathrm{L}$ has exactly one element equal to one. Rows of $\mathrm{L}$ whose element equal to 1 is in the same column will be labelled with the same single symbol (i.e., according to the definition in (\ref{app1eq30})), and thus we can see that this will correspond to the product of $\ket{\Phi^+}$ and $\ket{\mathrm{GHZ}}$ states as specified, thus completing the proof.

\section{Summary}

In this article, we consider the problem of entanglement distribution in quantum architectures with constraints on the interactions between pairs of qubits, described by a network $G$.
We describe how this problem may be fruitfully reduced to solving the $k$ pairs problem through linear network coding, on the same network $G$; and we describe how such codes may be simulated to achieve entanglement distribution using a shallow circuit, independent of the size of $G$ or the distance over which the entanglement is to be distributed.
We also present several novel observations about realising linear network codes through stabiliser circuits.

For the purposes of practically realising operations on practical quantum architectures, it will be of interest both to reduce the depth of circuits to distribute entanglement, and to efficiently discover protocols to do so.
However, it will also be important to address issues which we have not considered here, such as the fidelity of the entanglement which is distributed to some known Bell state.
We do not expect the quality of such Bell states to be independent of the distance over which such entangled states are distributed, or the number of Bell states which are distributed in parallel using QLNC circuits.
Nevertheless, it may be feasible to consider techniques to mitigate what noise may be present.
We hope that it may be possible to do so while incorporating the apparent benefits that QLNC circuits theoretically provide in the noiseless case.

\section*{Acknowledgements}
This work was supported by an Networked Quantum Information Technologies Hub Industrial Partnership Project Grant. The authors also thank Earl Campbell, Steve Brierley and the team at Riverlane for their encouragement and the general discussions that helped to shape this paper. 

\appendix

\end{document}